\documentclass[10pt,a4paper]{article}

\usepackage{amsmath,amsgen,latexsym}
\usepackage{amstext,amssymb,amsfonts,latexsym}
\usepackage{theorem}
\usepackage{pifont}
\usepackage{microtype}
\usepackage{graphicx}
\usepackage{url}

 \newcommand{\bs}{\bigskip}
 \newcommand{\ms}{\medskip}
 \newcommand{\n}{\noindent}
 \newcommand{\s}{\smallskip}
 \newcommand{\hs}[1]{\hspace*{ #1 mm}}
 \newcommand{\vs}[1]{\vspace*{ #1 mm}}

 \newcommand{\setempty}{\varnothing}
 
 \newcommand{\nat}{\mathbb{N}}
 
 \newcommand{\integer}{\mathbb{Z}}
 
 \newcommand{\complex}{\mathbb{C}}

 \newcommand{\etalc}{\textrm{et al.}}

 \newcommand{\CC}{{\cal C}}
 \newcommand{\FF}{{\cal F}}

 \newcommand{\HH}{{\cal H}}

 \newcommand{\GG}{{\cal G}}

 \newcommand{\PP}{{\cal P}}
 
 \newcommand{\UU}{{\cal U}}

 \newcommand{\dl}{\mathrm{L}}
 \newcommand{\nl}{\mathrm{NL}}
 \newcommand{\p}{\mathrm{P}}
 \newcommand{\np}{\mathrm{NP}}

 \newcommand{\fl}{\mathrm{FL}}
 \newcommand{\fp}{\mathrm{FP}}
 \newcommand{\sharpp}{\#\mathrm{P}}

\theoremstyle{plain}
\theoremheaderfont{\bfseries}
 \newtheorem{theorem}{Theorem}[section]
 \newtheorem{lemma}[theorem]{Lemma}
 \newtheorem{proposition}[theorem]{{\bf Proposition}}
 \newtheorem{corollary}[theorem]{Corollary}
 \newtheorem{fact}[theorem]{Fact.}

 {\theorembodyfont{\rmfamily}
  }
 {\theorembodyfont{\rmfamily} }
 {\theorembodyfont{\rmfamily} }

 \newtheorem{claim}{Claim}

 \newenvironment{proof}{\par \noindent
            {\bf Proof. \hs{2}}}{\hfill$\Box$ \vspace*{3mm}}

 \newenvironment{proofof}[1]{\vspace*{5mm} \par \noindent
         {\bf Proof of #1.\hs{2}}}{\hfill$\Box$ \vspace*{3mm}}

 \newcommand{\pair}[1]{\langle #1 \rangle}

\newcommand{\ignore}[1]{}

\newcommand{\sharpcfl}{\#\mathrm{CFL}}
\newcommand{\sac}{\mathrm{SAC}}

\newcommand{\logcfl}{\mathrm{LOGCFL}}
\newcommand{\sharplogcfl}{\#\mathrm{LOGCFL}}

\newcommand{\sharpl}{\mathrm{\#\mathrm{L}}}
\newcommand{\sharpsac}[1]{\#\mathrm{SAC}^{ #1 }}
\newcommand{\dlogtime}{\mathrm{DLOGTIME}}

\newcommand{\sharpcsp}{\#\mathrm{CSP}}
\newcommand{\sharpacsp}{\#\mathrm{ACSP}}
\newcommand{\supp}{\mathrm{supp}}
\newcommand{\auxpdatisp}[2]{\mathrm{AuxPDA},\!\mathrm{TISP}( #1, #2 )}
\newcommand{\nauxpdatisp}[2]{\mathrm{NAuxPDA},\!\mathrm{TISP}( #1, #2 )}

\newcommand{\nc}[1]{\mathrm{NC}^{ #1 }}
\newcommand{\IM}{\mathcal{IM}}
\newcommand{\ED}{\mathcal{ED}}
\newcommand{\DG}{\mathcal{DG}}
\newcommand{\NZ}{\mathcal{NZ}}

\newcommand{\boldvec}[1]{\mbox{\boldmath $ #1 $}}

\begin{document}

\pagestyle{plain}
\pagenumbering{arabic}
\setcounter{page}{1}
\setcounter{footnote}{0}

\begin{center}
{\Large {\bf Complexity Classification of
Complex-Weighted Counting Acyclic Constraint Satisfaction Problems}}\footnote{An extended abstract of this current article is
scheduled to appear in the Proceedings of the 12th Computing Conference,  London, UK, July 11--12, 2024, Lecture Notes in Networks and Systems, Springer-Verlag, 2024.}
\bs\s\\

{\sc Tomoyuki Yamakami}\footnote{Present Affiliation: Faculty of Engineering, University of Fukui, 3-9-1 Bunkyo, Fukui 910-8507, Japan}
\end{center}
\ms


\begin{abstract}
We study the computational complexity of counting constraint satisfaction  problems (\#CSPs) whose constraints assign complex numbers to Boolean inputs
when the corresponding constraint hypergraphs are acyclic. These problems are called acyclic \#CSPs  or succinctly, \#ACSPs.
We wish to determine the computational complexity of all such \#ACSPs
when arbitrary unary constraints are freely available.
Depending on whether we further allow or disallow the free use of the specific constraint XOR (binary disequality), we present two complexity classifications of the \#ACSPs according to the types of constraints used for the problems.
When XOR is freely available, we first obtain a complete dichotomy classification. On the contrary, when XOR is not available for free, we then obtain a trichotomy classification.
To deal with an acyclic nature of constraints in those classifications, we develop a new technical tool called acyclic-T-constructibility or AT-constructibility, and we exploit it to analyze a complexity upper bound of each \#ACSPs.

\ms

\n{\bf Key words:} auxiliary pushdown automata, \#LOGCFL, counting constraint satisfaction problem, \#ACSP, complexity classification, acyclic hypergraph, acyclic-T-constructibility
\end{abstract}

\sloppy
\section{A Historical Account and an Overview of Contributions}\label{sec:introduction}

We briefly review the background and the historical accounts of the main subject of counting satisfaction problems and then provide the major contributions of this work.

\subsection{Counting Constraint Satisfaction Problems or \#CSPs}\label{sec:counting-CSPs}

\emph{Constraint satisfaction problems} (or CSPs, for short) have played an important role in computer science. Typical CSP examples include the \emph{formula satisfiability problem} (SAT), which is well-known to be $\np$-complete. It is of great importance to determine the computational complexity of each CSP and classify all CSPs according to the types of constraints used as inputs to the CSPs. Schaefer \cite{Sch78} is the first to give a complete classification of (unweighted) CSPs with Boolean domain. Any of such CSPs falls into two categories. It is either in $\p$ or $\np$-complete. The classification of this type is called
a \emph{dichotomy classification} or a \emph{dichotomy theorem}.

Apart from decision problems, \emph{counting problems} are another important subject of computer science in theory and in practice.
In the 1970s, Valiant \cite{Val76,Val79} paid special attention to \emph{counting problems} (viewed as \emph{counting functions}), each of which computes the number of accepting computation paths of a nondeterministic Turing machines running in polynomial time. Those counting problems form the complexity class $\sharpp$.
There are a number of natural $\sharpp$-complete\footnote{For the notion of $\sharpp$-completeness, Valiant \cite{Val79} used Turing reductions instead of more common many-one reductions. In certain cases, his $\sharpp$-completeness holds also under many-one reductions. Such a case is $\#\mathrm{3SAT}$. \`{A}lvarez and Jenner \cite{AJ93}
remarked that the $\sharpp$-completeness of $\#\mathrm{3SAT}$ holds under logarithmic-space many-one reduction.} problems, including $\#\mathrm{2SAT}$ (counting 2CNF Boolean satisfiability problem), under polynomial-time Turing reduction \cite{Val79}.
Our primary interest of this work lies in \emph{counting constraint satisfaction problems} (abbreviated as \#CSPs), which count the number of assignments that satisfy all the constraints of target CSPs.
These \#CSPs have a close connection to certain physical phenomena because the \#CSPs can be viewed as ``partition functions'', which are essential parts of physical systems.
There has been a large volume of research promoting our understandings of \#CSPs.

The behaviors of \#CSPs, however, are quite different from the aforementioned CSPs, which we emphatically cite as \emph{decision CSPs} for clarity. In the case of such decision CSPs, for instance,
it is well-known that $\mathrm{3SAT}$ (3CNF SAT) is $\np$-complete while $\mathrm{2SAT}$ (2CNF SAT) is $\nl$-complete (under logarithmic-space many-one reductions).
In contrast, the counting versions of $\mathrm{3SAT}$ and $\mathrm{2SAT}$, denoted respectively by $\#\mathrm{3SAT}$ and $\#\mathrm{2SAT}$, were both proven in \cite{Val79} to be $\sharpp$-complete, where $\sharpp$ is a counting version of $\p$ introduced by Valiant \cite{Val76,Val79}.

In this work, we are particularly interested in \emph{weighted} \#CSPs with \emph{Boolean domain} 
whose \emph{constraints} refer to pairs of a sequence of variables and a function (called a constraint function) that assign ``weights'' to variable tuples.
The choice of weights of constraints significantly affect the computational complexity of $\sharpcsp$s.
Creignou and Herman \cite{CH96} first gave a complete classification of $\sharpcsp$s with $\{0,1\}$-valued constraint functions (or ``unweighted'' \#CSPs). This result was further extended to the case of weighted  constraints. For \emph{nonnegative-real-weighted \#CSPs}, for instance, Dyer, Goldberg, and Jerrum \cite{DGJ09} obtained their complexity classification whereas
Cai, Lu, and Xia \cite{CLX14} obtained a complexity classification of all   \emph{complex-weighted \#CSPs}. Here, real and complex numbers are treated as ``symbolic objects'' during any computations (see Section \ref{sec:treatment} for a brief explanation).
In fact, the complexity classification of Cai et al. asserts that, assuming that unary constraints are freely available as part of given constraints, any complex-weighted \#CSP is either in $\fp_{\complex}$ or $\sharpp$-hard, where $\fp_{\complex}$ denotes a complex-number extension of $\fp$ (polynomial-time computable function class) \cite{CLX14,Yam12a,Yam12b,Yam12c,Yam14}.

When \emph{randomized approximate counting} is concerned, instead of ``exact'' counting, the computational complexity of $\sharpcsp$s was discussed first by Dyer \etalc~\cite{DGJ10} and then by Yamakami \cite{Yam12a,Yam12b,Yam12c,Yam14}.
While Dyer et al. gave an approximation complexity classification of all unweighted \#CSPs, Yamakami showed in \cite{Yam12a} an approximation complexity classification of all complex-weighted \#CSPs.

Most studies on the complexity classifications of unweighted/weighted \#CSPs have been conducted in the polynomial-time setting.
It still remains challenging to expand the scope of these studies by looking into \#CSPs that are already tractable (i.e., solvable in polynomial runtime).
In the case of decision CSPs, nevertheless, an important milestone in this direction is the work of Gottlob, Leone, and Scarcello \cite{GLS01}, who first studied a natural restriction of CSPs, known as \emph{acyclic CSPs} (or succinctly, ACSPs), within the framework of
acyclic conjunctive queries in database theory \cite{BFMY83}. Here, a
CSP is called \emph{acyclic} if its corresponding constraint (undirected) hypergraph is  acyclic.
In particular, they presented a few concrete examples of ACSPs for
the complexity class $\logcfl$,
studied first by Cook \cite{Coo71},
which is the closure of all context-free languages under the logarithmic-space many-one reductions
(or $\dl$-m-reductions, for short)
and it is a quite robust complexity class, exhibiting various important characteristics of parallel computing, located between $\nc{1}$ and $\nc{2}$.
A simple example of the languages in $\logcfl$ is an acyclic restriction of 2SA, called the \emph{acyclic 2CNF satisfiability problem} ($\mathrm{Acyc\mbox{-}2SAT}$).

Sudborough \cite{Sud78} characterized $\logcfl$ in terms of \emph{two-way auxiliary nondeterministic pushdown automata} (or aux-2npda's, for short). We use the notation $\nauxpdatisp{t(n)}{s(n)}$ to denote the family of all decision problems (equivalently, languages) recognized by such machines that run in  time $O(t(n))$ using space $O(s(n))$. With this notation, Sudborough's result is expressed as $\logcfl =  \nauxpdatisp{n^{O(1)}}{\log{n}}$. Other important characterizations of $\logcfl$ include multi-head two-way nondeterministic pushdown automata \cite{Sud78}, alternating Turing machines  \cite{Ruz80}, and first-order logical formulas \cite{LMSV01}.
Venkateswaran \cite{Ven91}, for instance, demonstrated that the languages in $\logcfl$ are precisely computed by uniform families of semi-unbounded Boolean circuits of polynomial size and logarithmic depth. This exemplifies the importance and naturalness of $\logcfl$ in the field of computational complexity theory.

As counting versions of $\nauxpdatisp{t(n)}{s(n)}$ and $\mathrm{SAC}^1$, Vinay \cite{Vin91} and Niedermeier and Rossmanith \cite{NR95} studied $\#\auxpdatisp{t(n)}{s(n)}$ and $\#\sac^1$ and they showed that those counting complexity classes actually coincide. Similarly to $\sharpp$, we denote by $\sharpcfl$ the set of all functions that output the number of accepting computation paths of 1npda's. Reinhardt \cite{Rei92} claimed that $\#\auxpdatisp{n^{O(1)}}{\log{n}}$ coincides with $\mathrm{FLOG}(\sharpcfl)$, which is the class of all functions logarithmic-space many-one reducible to $\sharpcfl$. a counting analogue of $\logcfl$.

Along this line of research, we wish to study the computational complexity of a counting version of ACSPs, the \emph{counting acyclic CSPs} (or $\#\mathrm{ACSP}$s for short)
\emph{with complex-weighted constraints}
in hopes of making a crucial contribution to the field of CSPs.
This is the first study on the complexity classification of \#ACSPs.
For such a discussion on the complexity issues of  \#ACSPs, we need to work within the counting complexity class $\#\logcfl$, which is a counting analogue of $\logcfl$ (which will be defined in Section \ref{sec:basic-notion}).
Since $\sharplogcfl$ is included in $\#\mathrm{AC}^1$, the computations of \#ACSPs can be highly parallelized.
A counting variant of $\mathrm{Acyc\mbox{-}2SAT}$, denoted $\#\mathrm{Acyc\mbox{-}2SAT}$, is a typical example of \#ACSPs.
This counting problem counts the number of satisfying assignments of a given acyclic 2CNF Boolean formula (see Section \ref{sec:complexity-ACSP}) and
belongs to $\sharplogcfl$ but its $\sharplogcfl$-completeness is unknown at this moment. This situation sharply contrasts
the $\sharpp$-completeness of $\#\mathrm{2SAT}$  (under polynomial-time Turing reductions)  \cite{Val79}.

Since our computation model is much weaker than polynomial-time computation, we cannot use polynomial-time reductions. Throughout this work, we use logarithmic-space computation as a basis of our  reductions. Between two counting problems, we use a variant of logarithmic-space many-one reductions (or $\dl$-m-reductions) designed for decision problems, called in this paper \emph{logspace reductions}. Hence, the notion of ``completeness'' in this paper is limited to logspace reductions (stated formally in Section \ref{sec:machines}).

Here, we remark that $\sharplogcfl$ certainly has complete counting problems under \emph{logspace reductions} (see Section \ref{sec:basic-notion} for their formal definition), including the \emph{ranking of 1dpda problem}, denoted $RANK_{1dpda}$, in which we count the number of lexicographically smaller strings accepted by one-way deterministic pushdown automata (or 1dpda's) \cite{Vin91}.
Later in Section \ref{sec:basic-notion}, we will present another canonical $\sharplogcfl$-complete problem, called the \emph{counting SAC$^1$ problem} (or $\#\mathrm{SAC1P}$).

\subsection{Main Contribution of This Work}\label{sec:main-contribution}

Here, we wish to determine the computational complexity of each  complex-weighted \#ACSP, where complex numbers produced by constraint functions are generally treated as ``symbolic objects'' (see Section \ref{sec:treatment} for a more explanation), which are essentially different from ``computable numbers''.
In particular, under the common assumption that unary constraints are freely available as part of input constraints,
we present two complete complexity classifications of all complex-weighted \#ACSPs, as similarly done in \cite{CLX14,Yam12a,Yam12b,Yam12c},
depending on whether XOR (binary disequality) is further allowed or not.
To improve the readability, we tend to use the generic term ``constraints'' in place of ``constraint functions'' as long as no confusion occurs.

We can raise a natural question of whether there is a natural characteristic for \#ACSPs so that it captures a lower complexity class than $\sharpp$. We intend to classify all such \#ACSPs according to the types of their constraints. 
For clarity, we write $\sharpacsp(\FF)$ to denote the set of complex-weighted \#ACSPs whose constraints are all taken from $\FF$. 
We intend to study the computational complexity of such $\sharpacsp(\FF)$ for various choices of constraint functions that assign complex numbers to Boolean inputs.

To describe the complexity classifications of \#ACSPs, let us first review from \cite{Yam12a} a special constraint set, called $\ED$, which consists
of all constraints obtained by multiplying the equality (EQ) of arbitrary arity, the binary disequality (XOR), and any unary constraints (see Section \ref{sec:constraints} for their precise definitions).
Let $\UU$ denote the set of all unary constraints.
Assuming that $XOR$ and unary constraints are freely included in given constraints, we obtain the following form of dichotomy classification.
Similar to $\fp_{\complex}$, the notation $\fl_{\complex}$ refers to a complex-number extension of $\fl$ (log-space computable function class).

\begin{theorem}\label{main-theorem}
For any set $\FF$ of constraints, if $\FF\subseteq \ED$, then $\sharpacsp(\FF,\UU,XOR)$ belongs to $\fl_{\complex}$. Otherwise, it is $\sharplogcfl$-hard under logspace reductions.
\end{theorem}

On the contrary, when $XOR$ is not freely usable, we obtain the following trichotomy classification. Recall from \cite{Yam12a} the notation $\IM$, which denotes the set of all nowhere-zero-valued constraints obtained by multiplying the ``implication'' and any unary constraints (see Section \ref{sec:constraints} for their formal definitions).
Notice that the implication leads to the equality.

\begin{theorem}\label{second-theorem}
For any set $\FF$ of constraints, if $\FF\subseteq \ED$, then $\sharpacsp(\FF,\UU)$ belongs to $\fl_{\complex}$. Otherwise, if all constraints are in $\IM$, then it is hard for $\#\mathrm{Acyc\mbox{-}2SAT}$ under logspace reductions. Otherwise, it is $\sharplogcfl$-hard under logspace reductions.
\end{theorem}

As noted earlier, $\#\mathrm{Acyc\mbox{-}2SAT}$ is not yet known to be $\#\logcfl$-complete (under logspace reductions).
If $\#\mathrm{Acyc\mbox{-}2SAT}$ is indeed complete for $\#\logcfl$, then Theorem \ref{second-theorem} truly provides a dichotomy classification.
In contrast, if $\#\mathrm{Acyc\mbox{-}2SAT}$ is not $\sharplogcfl$-complete, then the collection of all integer-valued counting problems reducible to it forms a distinctive counting complexity class sitting in between $\sharpl$ and $\sharplogcfl$.

To prove the above two theorems, we will introduce a new technical tool, called \emph{acyclic T-constructibility} (abbreviated as \emph{AT-constructibility}) in Section \ref{sec:acyclic-const} based on acyclic hypergraphs induced from \#ACSPs.  This is an adaptation of so-called \emph{T-constructibility} developed in \cite{Yam12a,Yam12b,Yam12c}.

\subsection{Treatment of Complex Numbers in This Work}\label{sec:treatment}

As noted in Section \ref{sec:counting-CSPs}, it is known that the choice of weight types (e.g., natural numbers, rationals, reals, or complex numbers) of constraints tend to alter the complexity classification of \#CSPs. Over the years, it has become a custom to use complex weights in the study of \#CSPs and, in this work, we also follow this trend and intend to deal with complex-weighted constraints.
For the sake of the curious reader, we briefly comment on how we treat
``arbitrary''
complex numbers
(not limited to ``computable'' ones)
in this work.

Those numbers are treated as basic, symbolic ``objects'' and some of them are initially included as part of input instances. We follow the existing convention of \cite{CL11,CLX14,Yam12a,Yam12b,Yam12c}  for freely ``expressing'' and ``calculating'' the complex numbers, which are treated as symbolic  ``objects'' during any computations. In the field of \emph{algebraic computing}, in particular, a uniform Turing machine model of Blum, Shub, and Smale \cite{BSS89} has been used to work over an arbitrary field or ring $\mathbb{F}$.

Based on this machine model, we can freely manipulate the elements of $\mathbb{F}$ by conducting simple arithmetical operations, such as multiplication, addition, and division, in a clear and direct manner. In this work, $\mathbb{F}$ is set to be $\complex$.

For more background on this machine model and its induced complexity classes, such as $\p_{\complex}$ and $\np_{\complex}$, the interested reader should refer to, e.g., textbooks \cite[Section 6.3]{AB09} and \cite[Section 1.4]{CC17} as well as references therein.

\section{Preparation: Basic Notions and Notation}\label{sec:basic-notion}

We explain the terminology used in the rest of this work.

\subsection{Numbers, Sets, and Functions}

Two notations  $\integer$ and $\nat$ represent the sets of all \emph{integers} and of all \emph{natural numbers} (i.e., nonnegative integers), respectively. Given two numbers $m,n\in\integer$ with $m\leq n$, $[m,n]_{\integer}$ denotes the \emph{integer interval}  $\{m,m+1,m+2,\ldots,n\}$. In particular, when $n\geq1$, we abbreviate $[1,n]_{\integer}$ as $[n]$.
Let $\integer_2$ denote the group of $\{0,1\}$. We use the standard operations over $\integer_2$, including $AND_2$, $OR_2$, and $XOR$, where the subscript $2$ indicates the binary operations. We then define $Implies(x,y)$ to denote $OR_2(\bar{x},y)$, where $\bar{x}$ denotes the negation of $x$.
The notation $\complex$ denotes the set of all complex numbers. We then write $\imath$ for $\sqrt{-1}$.

In this work, any \emph{polynomial} must have nonnegative integer coefficients and any \emph{logarithm} is taken to the base $2$. A \emph{relation of arity $k$} (or a $k$-ary relation) over a set $D$ is a subset of $D^k$. Given a set $A$, its \emph{power set} is denoted $\PP(A)$. 

A finite nonempty set of ``symbols'' or ``letters'' is called an \emph{alphabet} and a finite sequence of an alphabet $\Sigma$ is a \emph{string} over $\Sigma$. The \emph{length} of a string $x$ is the total number of symbols in it and is denoted $|x|$. A \emph{language} over $\Sigma$ is a set of strings over $\Sigma$. The notation $\Sigma^*$ denotes the set of all strings over $\Sigma$.

Given an alphabet $\Sigma$, a function $f$ on $\Sigma^*$ (i.e., from $\Sigma^*$ to $\Sigma^*$) is said to be \emph{polynomially bounded} if there is a polynomial $p$ satisfying $|f(x)| \leq p(|x|)$ for all strings $x\in\Sigma^*$.
Similarly, a function $f$ from $\Sigma^*$ to $\nat$ is \emph{polynomially bounded} if a certain polynomial satisfies $f(x)\leq p(|x|)$ for all $x\in\Sigma^*$.

\subsection{Boolean Circuits and Hypergraphs}

We intend to work on (undirected) hypergraphs. A \emph{hypergraph} is of the form $(V,E)$ with a finite set $V$ of vertices and a set $E$ of \emph{hyperedges}, where a hyperedge is a subset of $V$.  A hyperedge is \emph{empty} if it is the empty set. The \emph{empty hypergraph} has no vertex (and thus has only the empty hyperedge).
A hypergraph $G=(V,E)$ is said to be \emph{acyclic} if, after applying the following actions (i)--(ii) finitely many times, $G$ becomes the empty hypergraph: (i) remove vertices that appear in at most one hyperedge and (ii) remove hyperedges that are either empty or contained in other hyperedges. This notion of acyclicity is also called \emph{$\alpha$-acyclicity}  \cite{Fag83}.
It is important to note from \cite{BFMY83} that the order of applications of those actions does not affect the form of the final hypergraph.
In \cite[Corollary 3.7]{GLS01}, it is shown that the problem of determining whether or not a given hypergraph is acyclic belongs to $\mathrm{SL}$ (symmetric NL). Since $\mathrm{SL}=\dl$ \cite{Rei08},
this problem actually falls into $\dl$.
Hence, we conclude the following.

\begin{fact}\label{acyclic-logspace}
Using logarithmic space, we can determine the acyclicity of hypergraphs.
\end{fact}

This fact will be implicitly used throughout this work.

A \emph{Boolean circuit} (or simply, a \emph{circuit}) is an acyclic directed graph whose internal nodes are called gates, nodes of indegree $0$ are called input gates, and leaves are output gates. 
As other gates, we use AND ($\wedge$), OR ($\vee$), and the negation ($\overline{x}$). As customary, the negation is applied only to input variables. The indegree and the outdegree of a gate are called its \emph{fan-in} and \emph{fan-out}. 
Let us consider a Boolean circuit $C$ of $n$ inputs and a binary string $x$ of length $n$. 
A circuit is \emph{leveled} if all gates are assigned to ``levels'' so that input gates are at level $0$, any gate at level $i$ take inputs from gates at level $i-1$. A leveled circuit is called \emph{alternating} if (i) all gates at the same level are of the same type ($\wedge$ and $\vee$) and (ii) gates at odd levels and gates at even levels have different types. A leveled alternating circuit is said to be \emph{semi-unbounded} if all AND gates in in have fan-in $2$ but we allow OR gates to have unbounded fan-in. In this work, all gates are assumed to have fan-out $1$. 

An \emph{accepting subtree} of $C$ on the input $x$ is a subtree $T$ of $C$ that satisfies the following: (i) all leaves of $T$ are evaluated by $x$, (ii) the root of $T$ is evaluated to be true, (iii) each AND gate $G$ in $T$ has two children in $T$ and if these children are evaluated to be true, then so is $G$, and (iv) each OR gate $G$ in $T$ has exactly one child in $T$ and if the child is evaluated to be true, then so is $G$.

To handle circuits as part of input instances, we use an appropriately defined encoding of a circuit $C$, which is composed of ``labels'' of all gates in $C$, where the \emph{label} of a gate consists of the information on its gate location (its level, gate number at each level), gate type ($\vee$, $\wedge$, $\neg$), and its direct connection to its child gates.

\subsection{Machine Models and Counting}\label{sec:machines}

As an underlying fundamental computation model, we use \emph{multi-tape deterministic Turing machines} (or DTMs) with read-only input tapes, rewritable work tapes, and (possibly) write-once\footnote{A tape is called \emph{write-once} if its tape head never moves to the left and, whenever the tape head writes down a non-blank symbol, it must move to the right.} output tapes. Let $\fl$ denote the collection of all functions from $\Sigma^*$ to $\Gamma^*$ for alphabets $\Sigma$ and $\Gamma$ that are  computable by DTMs in polynomial time using logarithmic space. 
Since the underlying DTMs run in polynomial time, these functions must be polynomially bounded.
A \emph{counting Turing machine} (or a CTM) is a nondeterministic Turing machine (NTM) whose outcome is the number of accepting computation paths. Based on CTMs, $\sharpl$ is defined to be composed of all functions witnessed by polynomial-time logarithmic-space (or log-space) CTMs with read-only input tapes  \cite{AJ93}.

A \emph{one-way nondeterministic pushdown automaton} (or a 1npda, for short) is a tuple $(Q,\Sigma,{\{\triangleright,\triangleleft\}}, \Gamma,\delta,q_0,\bot,Q_{acc},Q_{rej})$.   A \emph{one-way counting pushdown automaton} (or a \#1pda) is fundamentally a 1npda but it is designed to ``output'' the number of accepting computation paths.
Similar to $\sharpl$, the function class $\sharpcfl$ is composed of all functions from $\Sigma^*$ to $\nat$ for arbitrary alphabets $\Sigma$ witnessed by
\#1pda's running in polynomial time.

A \emph{two-way nondeterministic auxiliary pushdown automaton} (or an aux-2npda, for short) is a polynomial-time two-way nondeterministic pushdown automaton equipped further with a two-way rewritable $O(\log{n})$ space-bounded auxiliary work tape.
Formally, an aux-2npda is expressed as a tuple $(Q,\Sigma,{\{\triangleright,\triangleleft\}}, \Gamma, \Theta,  \delta,q_0,\bot,Q_{acc},Q_{rej})$, where $Q$ is a finite set of (inner) states, $\Sigma$ is an input alphabet, $\Theta$ is an auxiliary (work) tape alphabet, $\Gamma$ is a stack alphabet, $q_0$ ($\in Q$) is the initial state, $\bot$ ($\in\Gamma$) is the bottom marker, $Q_{acc}$ and $Q_{rej}$ are sets of accepting and rejecting states, respectively, and $\delta$ is a transition function mapping $(Q-Q_{halt})\times \check{\Sigma}\times \Gamma \times \Theta$ to $\PP(Q\times\Gamma^*\times \Theta\times D_1\times D_2)$ with $\check{\Sigma} = \Sigma\cup\{\triangleright,\triangleleft\}$, $D_1=D_2=\{-1,0,+1\}$, and $Q_{halt} = Q_{acc}\cup Q_{rej}$.
The complexity class $\logcfl$ is  characterized in terms of aux-2npda's \cite{Sud78}.
In a straightforward analogy to decision problems, a counting variant of $\logcfl$ has been discussed in the past literature.

Similar to CTMs, a \emph{two-way counting auxiliary pushdown automaton} (or a \#aux-2pda) is an aux-2npda but it can output the total number of accepting computation paths on each input string. We write $\#\mathrm{AuxPDA,\!TISP}(n^{O(1)},\log{n})$ to denote the collection of all functions computed by \#aux-2pda's in time $n^{O(1)}$ using space $O(\log{n})$, where ``TI'' and ``SP'' respectively refer to ``time'' and ``space''.

The complexity class $\sharpsac{1}$ consists of all functions computing the total number of accepting subtrees of Boolean circuits taken from uniform families of semi-unbounded Boolean circuits of polynomial size and logarithmic depth. A family $\{C_n\}_{n\in\nat}$ of Boolean circuits is called \emph{L-uniform} if there exists a log-space DTM $M$ running in time polynomial in $n$ such that, for any input $1^n$, $M$ produces an encoding of $C_n$.
It is known that, for $\sac^{1}$, $\dl$-uniformity and $\dlogtime$-uniformity are interchangeable.


In late 1970s, Valiant \cite{Val76,Val79} used polynomial-time (Turing) reduction in a discussion on $\sharpp$-completeness using oracle Turing machines, which are DTMs equipped further with write-once query tapes and  read-only answer tapes.

Krentel \cite{Kre88} considered a \emph{polynomial-time metric reduction}  from $f$ to $g$, which is a pair $(h_1,h_2)$ of functions in $\fl$ with $h_1:\Sigma^*\to\Gamma^*$ and $h_2:\Sigma^*\times \nat\to\nat$ such that $f(x)=h_2(x,g(h_1(x)))$ for any $x\in\Sigma^*$. 
\`{A}lvarez and Jenner \cite{AJ93} used a log-space variant to show the $\sharpl$-completeness of counting problems. 

To compare the computational complexities of two functions $f$ and $g$, we here take the following functional analogue of the standard many-one reduction between two languages. We set the notation $f\leq^{\dl} g$ ($f$ is \emph{logspace reducible to} $g$) to indicate that
there exists a polynomially-bounded function $h\in\fl$ (called logspace reduction function) satisfying $f(x) = g(h(x))$ for any string $x\in\Sigma^*$ \cite{AJ93}.
Given a set $\FF$ of functions, we define $\mathrm{FLOG}(\FF)$ as the  closure of $\FF$ under logspace reductions.
As a special case, we take $\FF=\sharpcfl$ and consider $\mathrm{FLOG}(\sharpcfl)$.

Niedermeier and Rossmanith \cite{NR95} showed that $\#\sac^1$ equals $\#\mathrm{AuxPDA,\!SPTI}(\log{n},n^{O(1)})$. Reinhardt \cite{Rei92} claimed that $\mathrm{FLOG}(\sharpcfl)$ coincides with $\#\mathrm{AuxPDA,\!SPTI}(\log{n},n^{O(1)})$.
In summary, those three complexity classes coincide.

\begin{lemma}\label{basic-character}
{\rm \cite{NR95,Rei92}} $\#\mathrm{AuxPDA,\!TISP}(n^{O(1)},\log{n}) = \mathrm{FLOG}(\sharpcfl) = \#\sac^1$.
\end{lemma}

In analogy to $\sharpp$, in this work, we use the simple notation $\#\logcfl$ to express any of the complexity classes stated in the above lemma. 
Following this characterization of $\#\logcfl$, it is easy to show that $\sharplogcfl$ is closed under $\leq^{\dl}$.

\begin{lemma}
If $f\leq^{\dl} g$ and $g\in\sharplogcfl$, then $f\in\sharplogcfl$.
\end{lemma}

\begin{proof}
Assume that $f\leq^{\dl} g$ and $g\in\sharplogcfl$. Since $f\leq^{\dl} g$, we take a function $h\in\fl$ satisfying $f(x) = g(h(x))$ for all $x$. 
Since $g\in\sharplogcfl= \mathrm{FLOG}(\sharpcfl)$, there are two functions $d\in\sharpcfl$  and $k\in\fl$ such that $g(y)= d(k(y))$ for all $y$. 
We then obtain $f(x) = d(k(h(x)))$. It thus suffices to define $m=k\circ h$, which obviously belongs to $\fl$. It also follows by the definition that $f(x)=d(m(x))$ for all $x$. This implies that $f$ is in $\#\logcfl$, as requested.
\end{proof}

A function $f$ is \emph{$\sharplogcfl$-hard} (or hard for $\sharplogcfl$) under logspace reductions if $g\leq^{\dl}f$ holds for any $g\in\sharplogcfl$. If $f$ further satisfies $f\in\sharplogcfl$, then $f$ is called \emph{$\sharplogcfl$-complete} (or complete for $\sharplogcfl$).

Vinay \cite{Vin91} presented a counting problem complete for $\sharplogcfl$, the \emph{ranking of 1dpda problem} (abbreviated as $RANK_{1dpda}$). The \emph{rank} of a string $x$ in a language $L$ is the number of strings lexicographically smaller than $x$ in $L$. This problem $RANK_{1dpda}$ computes the rank of $x$ in $L(M)$ from a given encoding $\pair{M}$ of a 1dpda $M$ and a given string $x\in\{0,1\}^*$.

Based on Lemma \ref{basic-character}, we introduce another counting problem, called the \emph{counting SAC$^1$ problem} based on semi-unbounded circuits. We assume an efficient encoding $\pair{C}$ of a Boolean circuit $C$ into an appropriate binary string.

\s
{\sc Counting SAC$^1$ Problem} ({\sc \#SAC1P}):
\renewcommand{\labelitemi}{$\circ$}
\begin{itemize}\vs{-1}
  \setlength{\topsep}{-2mm}%
  \setlength{\itemsep}{1mm}%
  \setlength{\parskip}{0cm}%

\item {\sc Instance:} an encoding $\pair{C}$ of a leveled, alternating, semi-unbounded Boolean circuit of size at most $n$ and of depth at most $\log{n}$ with $n$ input bits and an input string $x\in\{0,1\}^n$.

\item {\sc Output:} the number of accepting subtrees of $C$ on $x$.
\end{itemize}\vs{-1}

This problem $\#\mathrm{SAC1P}$ will serve as a canonical $\sharplogcfl$-complete problem
in Section \ref{sec:complexity-ACSP}.

\begin{proposition}\label{SAC1P-complete}
$\#\mathrm{SAC1P}$ is $\sharplogcfl$-complete under logspace reductions.
\end{proposition}

\begin{proof}
Firstly, we show that $\#\mathrm{SAC1P}$ belongs to $\sharplogcfl$.
Let $\pair{\pair{C},x}$ denote any instance given to $\#\mathrm{SAC1P}$ and set $n=|x|$. By the input requirement of $\#\mathrm{SAC1P}$, $C$ is a leveled, alternating, semi-unbounded Boolean circuit of size at most $n$ and of depth at most $\log{n}$. Note that $|\pair{C}|=O(n)$ and $|\pair{\pair{C},x}|=O(n)$.
Similar to the proof of $\#\mathrm{SAC}^1 \subseteq \sharplogcfl$ \cite{NR95}, we can simulate
$C(x)$ by an appropriate log-space auxiliary pushdown automaton in polynomial time.

Secondly, since $\#\mathrm{SAC}^1=\sharplogcfl$, it suffices to show that every counting problem $P$ in $\#\mathrm{SAC}^1$ is logspace reducible to $\#\mathrm{SAC1P}$. Let $P$ be any counting problem in $\#\mathrm{SAC}^1$.
We write $P(x)$ for the output value of $P$ on instance $x$.
Let us take an $\dl$-uniform family $\CC=\{C_n\}_{n\in\nat}$ of leveled semi-unbounded Boolean circuits of polynomial size and logarithmic depth such that, for any binary string $x$, $P(x)$ equals the number of accepting subtrees of $C_{|x|}$ on $x$.
Take a fixed constant $k\in\nat^{+}$ such that the depth of $C_n$ is at most $k\log n$ and the size of $C_n$ is at most $n^k$.
We then define $f(x)= \pair{\pair{C_{|x|}},x}$ for any $x$. The $\dl$-uniformity of $\CC$ implies that $f$ is computable in polynomial time using $O(\log{n})$ space. Moreover, by the definition of $\CC$ and $f$, it follows that $P(x)$ equals $\#\mathrm{SAC1P}(f(x))$. Therefore, $P$ is logspace reducible to $\#\mathrm{SAC1P}$.
\end{proof}

\section{Counting Acyclic Constraint Satisfaction Problems or \#ACSPs}

We formally introduce complex-weighted counting acyclic CSPs (or \#ACSPs) and a technical tool called acyclic-T-constructibility.
We also prove basic properties of \#ACSPs.

\subsection{Various Constraints and Sets of Constraint Functions}\label{sec:constraints}

Since we are concerned with \emph{counting constraint satisfaction problems} (or \#CSPs, for short), we define a \emph{constraint} as a pair of the form $(f,(v_{1},v_{2},\ldots,v_{k}))$, where $f$ is a $k$-ary function and  $(v_{1},\ldots,v_{k})$ is a $k$-tuple of variables. We use the notation $arity(f)$ to denote this number $k$.
In general, a triplet $I=(Var,D,B,C)$ is called an instance of a \#CSP  if $Var=\{v_i\}_{i\in[t]}$ is a set of variables for $t\in\nat^{+}$, $D$ is a finite set, called a \emph{domain}, and $C=\{C_i\}_{i\in[s]}$, called a \emph{constraint set}, is a finite set of constraints of the form $C_i=(f_i,(v_{i_1},v_{i_2},\ldots,v_{i_{k_i}}))$ with  $f_i:D^k\to B$ for $s\in\nat^{+}$ and $v_{i_j}$ ranging over $D$. 
We call $f_i$ the \emph{$k_i$-ary constraint function} of $C_i$  and the value of $f_i$ are often called ``weights''.
Let $var(C_i)$ denote the set $\{v_{i_1},v_{i_2},\ldots,v_{i_{k_i}}\}$.
Remark that constraint functions are called ``signatures'' in \cite{CL11,CLX14,Yam12a,Yam12b,Yam12c}. 

An \emph{assignment} $\sigma$ is a function from $Var$ to $D$. Given such an assignment $\sigma:Var\to D$, we can evaluate the value of each $C_i$ by simply calculating $f_i(\sigma(v_{i_1}),\sigma(v_{i_2}),\ldots,\sigma(v_{i_{k_{i}}}))$. 
The \emph{support}\footnote{This notion was called in \cite{Yam12a} the ``underlying relation''.} $\supp(C)$ of a constraint $C=(f,(v_1,v_2,\ldots,v_k))$ is the set  $\{(d_1,d_2,\ldots,d_k)\in D^n \mid f(d_1,d_2,\ldots,d_k)\neq0\}$. When $D$ is the binary set $\{0,1\}$, a \#CSP is called \emph{Boolean}.

Throughout this work, we limit our interest to complex-weighted Boolean \#CSPs (namely, $B=\complex$ and $D=\{0,1\}$).
For the sake of simplicity, hereafter, we intentionally drop the term ``Boolean'' as well as ``complex-weighted''.
Since we discuss only complex-weighted Boolean constraints, the terse term ``constraint'' should be understood in this sense in the rest of this work unless stated otherwise.

Moreover, we omit $B$ and $D$ from $I$ and we succinctly write $(Var,C)$ in the rest of this work.
With a given \#CSP instance $I=(Var,C)$, we associate it with a labeled  hypergraph $G_I=(V_I,E_I)$, where $V_I=Var$ and $E_I=\{ \{v_{1},v_{2},\ldots,v_{k}\} \mid (f,(v_{1},v_{2},\ldots,v_{k}))\in C\}$,  whose hyperedge $\{v_{1},v_{2},\ldots,v_{k}\}$ has $f$ as its label.

We conveniently call $G_I$ the \emph{constraint hypergraph}\footnote{It is important to note that this association of $I$ to $G_I$ is \emph{not} unique because the order of the elements $v_{1},v_{2},\ldots,v_{k}$ is irrelevant for the hyperedges of $E_I$.
This is used here only for an introduction of our notion of ``acyclicity'' to utilize the results of \cite{GLS01}. A further extension of this  ``acyclicity'' notion may be possible based on a certain refinement of hypergraphs but we do not delve into this aspect in this work.} of $I$.
A \#CSP instance $I=(Var,C)$ is \emph{acyclic} if its constraint  hypergraph $G_I$ is acyclic.
An \emph{acyclic \#CSP} refers to a \#CSP whose instances are all restricted to acyclic ones.
For simplicity, we abbreviate acyclic \#CSPs as \emph{\#ACSPs}.

Given an instance $(Var,C)$ of a \#CSP, our purpose is to compute the complex value $count(I) = \sum_{\sigma}\prod_{i\in[s]} f_i(\sigma(v_{i_1}),\sigma(v_{i_2}),\ldots,\sigma(v_{i_{k_i}}))$, where   $\sigma$ ranges over all (global) assignments from $Var$ to $\{0,1\}$. In what follows, we freely identify $C_i$ with $f_i$ (associated with its underlying input variables) as long as $(\sigma(v_{i_1}),\sigma(v_{i_2}),\ldots,\sigma(v_{i_{k_i}}))$ is clear from the context. In this case, we loosely refer to $f_i$ as a ``constraint'', instead of a ``constraint function''.

We remark that, when all $f_i$'s are $\{0,1\}$-valued functions, we obtain the ``standard'' constraint satisfaction problems (CSPs). These CSPs are closely related to \emph{(rule-based) Boolean conjunctive queries} on database schemas and an acyclic variant of conjunctive queries has been intensively
studied in database theory (see, e.g., \cite{BFMY83,Fag83,GLS01}). We briefly call by ACSPs the acyclic variants of CSPs.


We assume the standard lexicographic order on $\{0,1\}^k$ for each fixed number $k\in\nat^+$. Each constraint function $f:\{0,1\}^k\to\complex$ can be expressed as a series of its output values according to this order on $\{0,1\}^k$. For instance, a binary constraint function $f$ is expressed succinctly as $(f(00),f(01),f(10),f(11))$.
With this expression, we obtain $OR_2=(0,1,1,1)$, $AND_2=(0,0,0,1)$, and $XOR=(0,1,1,0)$.
A \emph{symmetric constraint} is a constraint whose constraint function behaves depending only on the Hamming weights of inputs; namely, $f(x_1,x_2,\ldots,x_k) = f(x_{\pi(1)},x_{\pi(2)},\ldots,x_{\pi(k)})$ for any permutation $\pi:[k]\to[k]$.
The other constraint functions are called \emph{asymmetric}. For such a symmetric constraint function $f$ of arity $k$, we write $[a_0,a_1,\ldots,a_k]$ when $a_i$ is the value of $f(x)$ for all $x$ containing exactly $i$ 1s.
This notation helps us express $OR_2=[0,1,1]$, $AND_2=[1,0,0]$, and $XOR=[0,1,0]$. See \cite{CL11,CLX14,Yam12a} for more information.

Let us introduce several important sets of constraint functions.
Define $\Delta_0=[1,0]$ and $\Delta_1=[0,1]$, namely, for any constant $b\in\{0,1\}$, $\Delta_b(x)=1$ (resp., $=0$) if $x=b$ (resp., $x\neq b$).
Let $\mathcal{AND}=\{AND_k\mid k\geq2\}$ with $AND_k=[0,0,\ldots,0,1]$ ($k$ zeros) and let $\mathcal{OR} = \{OR_k\mid k\geq 2\}$ with $OR_k = [0,1,1,\ldots,1]$ ($k$ ones).
Moreover, let $NAND_k=[1,1,\cdots,1,0]$ ($k$ ones). Note that $XOR(x,y) = OR_2(x,y) NAND_2(x,y)$.
We also define $Implies=(1,1,0,1)$, which corresponds to the logical expression of ```$x\to y$''.
For convenience, the ``\emph{reverse implies}'' $RImplies$ is also introduced as $(1,0,1,1)$, which indicates ``$x\leftarrow y$''.
The \emph{equality} $EQ_{k}$ of arity $k$ is the function $f$ satisfying that $f(x_1,x_2,\ldots,x_k)=1$ iff $x_1=x_2=\cdots = x_k$. Note that $EQ_2(x,y) = Implies(x,y) RImplies(x,y)$.
The \emph{disequality} $NEQ_{k}$ of arity $k$ is the function $g$ such that $g(x_1,x_2,\ldots,x_k)= 1$ iff $EQ_{k}(x_1,x_2,\ldots,x_k)=0$.
In particular, when $k=2$, $NEQ_{2}$ is equivalent to $XOR$. Finally, we define $\UU$ to denote the set of all unary constraints. Note that $\Delta_0,\Delta_1\in \UU$ holds.

\subsection{Counting Problems \#ACSP($\FF$)}\label{sec:Sharp-ACSP}

We formally define the counting problem \#ACSP($\FF$), where $\FF$ is a fixed set of complex-valued constraint functions.

\s
{\sc $\FF$-Restricted Counting Acyclic Constraint Satisfaction Problem} ({\sc
\#ACSP($\FF$)}):
\renewcommand{\labelitemi}{$\circ$}
\begin{itemize}\vs{-1}
  \setlength{\topsep}{-2mm}%
  \setlength{\itemsep}{1mm}%
  \setlength{\parskip}{0cm}%

\item {\sc Instance:} $I=(Var,C)$ with a set $Var=\{v_i\}_{i\in[t]}$ of Boolean variables
    and a constraint set $C=\{C_i\}_{i\in[s]}$ with $C_i=(f_i,(v_{i_1},v_{i_2}, \ldots,v_{i_k}))$ and $f_i\in\FF\cup\{\Delta_0,\Delta_1\}$ for $s,t\in\nat^{+}$, provided that $I$ is acyclic.

\item {\sc Output:} the complex value $count(I) = \sum_{\sigma}\prod_{i\in[s]} f_i(\sigma(v_{i_1}),\sigma(v_{i_2}),\ldots,\sigma(v_{i_{k_i}}))$, where $\sigma$ ranges over all assignments from $Var$ to $\{0,1\}$. We call the value $\prod_{i\in[s]} f_i(\sigma(v_{i_1}),\sigma(v_{i_2}),\ldots,\sigma(v_{i_{k_i}}))$ a \emph{partial weight} of $I$ by $\sigma$.
\end{itemize}\vs{-1}

Due to the form of the above definition, we implicitly exclude $\{\Delta_0,\Delta_1\}$ from $\FF$. The acyclicity of $I$ helps us compute $count(I)$ on counting auxiliary automata. We respectively denote by $\sharpl_{\complex}$ and $\sharplogcfl_{\complex}$ the complex-number extensions of the counting complexity classes $\sharpl$ and $\sharplogcfl$ by allowing arbitrary complex numbers (see Section \ref{sec:treatment} for the treatment of complex numbers).

\begin{theorem}\label{falling-sharplogcfl}
For any constraint set $\FF$, $\sharpacsp(\FF)$ is in $\sharplogcfl_{\complex}$.
\end{theorem}


For a given constraint $C=(c,(v_1,v_2,\ldots,v_k))$, its \emph{local assignment} is a function $\pi_C$ from $var(C)$ ($=\{v_1,v_2,\ldots,v_k\}$) to $\{0,1\}$. For two local assignments $\pi_{C_1}$ and $\pi_{C_2}$ respectively for constraints $C_1$ and $C_2$, we say that $\pi_{C_1}$ \emph{agrees} with $\pi_{C_2}$ if, for any variable $x\in dom(\pi_{C_1})\cap dom(\pi_{C_2})$,
$\pi_{C_1}(x)=\pi_{C_2}(x)$ follows, where $dom(f)$ expresses the \emph{domain} of a function $f$.


A \emph{join forest} for a \#CSP instance $I=(Var,C)$, denoted $JF(I)$, is a forest whose vertices are constraints of $I$ and satisfies the following condition: if two constraints $C_1$ and $C_2$ share a common variable $v$, then $C_1$ and $C_2$ are connected and this $v$ must occur on a unique path between $C_1$ and $C_2$ (see, e.g., \cite{GLS01}). When a join forest has a single connected component, we particularly call it a \emph{join tree}.

In a similar way as in \cite{GLS01}, for a \#ACSP instance $I$, we further introduce the notion of \emph{weighted relational graph} $WG(I)$ whose vertices are all constraints in $I$ and edges are all unordered pairs $\{C_1,C_2\}$ of constraints in $I$ satisfying that (i) $C_1\neq C_2$ and (ii) both $C_1$ and $C_2$ share a common variable. Each edge $\{C_1,C_2\}$ is weighted by $w(C_1,C_2)$, which denotes the total number of common variables occurring in both $C_1$ and $C_2$. Consider a \emph{maximal-weight spanning forest} of $WG(I)$. Due to \cite{BG81} (noted in \cite[Lemma 3.1]{GLS01}), this forest is actually a join forest of $I$. We denote it by $JF(I)$. It follows from \cite{GLS01} (induced from \cite{BG81}) that we can determine whether or not a given unordered pair $\{C_1,C_2\}$ of constraints of an ACSP is an edge of $JF(I)$ using log space. The proof of Theorem \ref{falling-sharplogcfl} that follows below relies on this fact.

\vs{-2}
\begin{proofof}{Theorem \ref{falling-sharplogcfl}}
In the special case of $\FF=\setempty$, the theorem is obviously true. In what follows, we assume that $\FF\neq\setempty$. Consider any instance $I=(Var,C)$ of $\sharpacsp(\FF)$. Let $C=\{C_i\}_{i\in[s]}$ with $C_i=(f_i,(v_{i_1},v_{i_2},\ldots,v_{i_{k_i}}))$ for a certain number $t\in\nat^{+}$. We want to compute $count(I)$ using the associated \emph{edge-weighted relational graph} whose vertices are constraints of $I$ and edges are all unordered pairs $\{C_i,C_j\}$ of constraints satisfying that $C_i\neq C_j$ and $var(C_i)\cap var(C_j)\neq\setempty$, where $|var(C_i)\cap var(C_j)|$ is the weight of this edge. Consider its maximal-weight spanning forest $T$. 
This spanning forest $T$ becomes a join forest of $I$.
As noted above, using log space, it is possible to check if a given unordered pair of constraints is indeed an edge of $T$.

Meanwhile, we assume that $T$ is a tree. It suffices to construct an aux-2npda for which each computation path calculates $\prod_{i\in[s]} f_i(\sigma(v_{i_1}),\sigma(v_{i_2}),\ldots,\sigma(v_{i_{k_i}}))$, where $(f_i,(v_{i_1},v_{i_2},\ldots,v_{i_{k_i}}))$'s are constraints that appear on this path.
To compute $count(I)$, we traverse and backtrack all edges of $T$ from the root to leaves and backward by a depth-first, left-to-right manner. For this purpose, we use a stack to store the last node $C$ that has been visited, the choice of $\pi_{C}$, and the product of the values of constraint functions evaluated by the chosen local assignments.
In quest of the value $count(I)$, we need to exploit the fact that $G_I$ is an acyclic hypergraph.

The next claim helps us deal only with local assignments rather than a global assignment.
The claim can be proven in a way similar to the proof of \cite[Claim A]{GLS01} and we omit it.

\begin{claim}\label{claim-A-GLS01}
There exists an assignment $\sigma$ satisfying that $\prod_{i\in[s]} f_i(\sigma(v_{i_1}),\sigma(v_{i_2}), \ldots,\sigma(v_{i_{k_i}}))\neq0$ iff there exist $s$ local assignments $\pi_1,\pi_2,\ldots,\pi_s$ such that (i) each $\pi_i$ is a local assignment for $C_i$, (ii) for all indices $i\in[s]$, $f_i(\pi_i(v_{i_1}),\pi_i(v_{i_2}), \ldots,\pi_i(v_{i_{k_i}}))\neq0$, (iii) for all $i\in[2,s]_{\integer}$, $\pi_i$ agrees with $\pi_{p(i)}$, where $p(i)$ is in $[s]$ and $C_{p(i)}$ is a parent of $C_i$ in $JF(I)$.
\end{claim}

In visiting a new node $C$ from the last node $D$, we nondeterministically choose $\pi_C$ and check if $\pi_C$ agrees with $\pi_D$. If not, we reject immediately. Otherwise, we continue the traverse. If we reach the rightmost leaf, then we terminate the traverse and output the product of the values of constraint functions. This procedure is implemented on an appropriate log-space aux-2npda running in polynomial time. By the above claim, the nondeterministic choices of local assignments cannot contradict an associated global assignment. Therefore, the sum of all output values of our simulation matches $count(I)$.
\end{proofof}

\subsection{Acyclic-T-Constructibility (or AT-Constructibility)}\label{sec:acyclic-const}

A crucial technical tool used in \cite{Yam12a,Yam12b,Yam12c,Yam14} is various versions of so-called \emph{$T$-constructibility}, whose formulations were motivated by the earlier notions of ``implementation''  \cite{DGJ09} and ``gadget construction''  \cite{CLX14} used for the complexity classifications of \#CSPs. Since we deal with an acyclic restriction of $\sharpcsp$s, we need to introduce another variant of this technical tool for our later analyses of constraints.

Let $\GG$ denote any set of constraint functions and let $(f,(x_1,x_2,\ldots,x_k))$ be any constraint whose function $f$ is in $\GG$, where $(x_1,x_2,\ldots,x_k)$ is a $k$-tuple of Boolean variables associated with $f$.
Let $G=(V,E)$ be any hypergraph with $V=\{x_1,x_2,\ldots,x_k,y_1,y_2,\ldots,y_m\}$, where only vertices $x_1,x_2,\ldots,x_k$ may be associated with dangling edges.\footnote{A \emph{dangling edge} is treated as a special hyperedge consisting of exactly two elements such that one end of it is not connected to any vertex.}
We say that $G$ \emph{realizes} $(f,(x_1,x_2,\ldots,x_k))$ by $\GG$ if
there exist a constant $\lambda\in\complex-\{0\}$ and a finite set $B\subseteq \GG$ satisfying that
$f(\sigma(x_1),\sigma(x_2),\ldots,\sigma(x_k)) = \lambda \cdot \sum_{\tau} \prod_{g} g(\tilde{\sigma}(z_1),\tilde{\sigma}(z_2), \ldots, \tilde{\sigma}(z_d))$ for any assignment $\sigma:V-\{y_1,y_2,\ldots,y_m\}\to\{0,1\}$, where $g$ ranges over the constraint functions in $B$ satisfying $\{z_1,z_2,\ldots,z_d\}\in E$, $\tau$ ranges over all assignments from $\{y_1,y_2,\ldots,y_m\}$ to $\{0,1\}$, and $\tilde{\sigma}(u)=\tau(u)$ if $u\in \{y_1,y_2,\ldots,y_m\}$ and $\tilde{\sigma}(u)=\sigma(u)$ otherwise.
For ease of later descriptions, we omit $\sigma$ and $\tau$ altogether and simply write $f(x_1,x_2,\ldots,x_k) = \lambda \cdot \sum_{y_1,y_2,\ldots,y_m\in\{0,1\}} \prod_{g} g(z_1,z_2,\ldots,z_d)$, assuming that we identify the symbols $x_1,x_2,\ldots,x_k,z_1,z_2,\ldots,z_d$ with their actual Boolean values.
Such a hypergraph is called a \emph{$\GG$-realizable constraint hypergraph} of $(f,(x_1,x_2,\ldots,x_k))$.
A constraint $(f,(x_1,x_2,\ldots,x_k))$ is said to be \emph{acyclic-$T$-constructible} (or AT-constructible) from $\GG$ (denoted $f\leq^{acyc}_{con}\GG$) if there exists an acyclic hypergraph $G$ that realizes $f$ by $\GG$; in other words, there is an acyclic $\GG$-realizable constraint hypergraph of $(f,(x_1,x_2,\ldots,x_k))$. As customary, when $(x_1,x_2,\ldots,x_k)$ is clear from the context, we omit it and succinctly write $f$ in place of $(f,(x_1,x_2,\ldots,x_k))$.

Let us see some quick examples. For any number $k\geq2$, it follows that $AND_k\leq^{acyc}_{con} \{EQ_k,\Delta_1\}$ because $AND_k(x_1,x_2,\ldots,x_k) = EQ_k(x_1,x_2,\ldots,x_k) \Delta_1(x_1)$. Moreover, $EQ_3$ is $AT$-constructible from $EQ_2$ as $EQ_3(x,y,z) = EQ_2(x,y) EQ_2(y,z)$. Since $Implies(x,y) = \sum_{z\in\{0,1\}} OR_2(z,y) XOR(x,z)$, we obtain $Implies\leq^{acyc}_{con}\{OR_2,XOR\}$. From $EQ_2(x,y)=Implies(x,y)Implies(y,x)$ follows $EQ_2\leq^{acyc}_{con} \{Implies\}$.
When $XOR(x,y)$ is expressed as $\sum_{z\in\{0,1\}} OR_2(x,z) OR_2(y,z) OR_2(x,y) u_1(z)$ with $u_1=[-1,1]$, the corresponding realizable constraint hypergraph is not acyclic. 
By contrast, it also follows that $EQ_2(x,y)= \sum_{w\in\{0,1\}} \sum_{z\in\{0,1\}} OR_2(w,x) OR_2(z,y) XOR(y,w)XOR(x,z)$ and that the associated realizable constraint hypergraph is cyclic.


Let us recall from \cite{Yam12a} five useful operations over constraint functions: normalization, expansion, pinning, projection, and linking.

\vs{-2}
\renewcommand{\labelitemi}{$\circ$}
\begin{enumerate}\vs{-1}
  \setlength{\topsep}{-2mm}%
  \setlength{\itemsep}{1mm}%
  \setlength{\parskip}{0cm}%

\item \emph{Normalization} is an operation of generating from $f$ the constraint $\lambda \cdot f$ for a chosen constant $\lambda\in\complex-\{0\}$ defined as $(\lambda\cdot f)(x_1,x_2,\ldots,x_k) = \lambda \cdot f(x_1,x_2,\ldots,x_k)$.

\item \emph{Expansion} is an operation by which we introduce a new variable, say, $y$, choose an index $i\in[k]$, and generate $f'(x_1,\ldots,x_i,y,x_{i+1},\ldots,x_k) = f(x_1,\ldots,x_i,x_{i+1},\ldots,x_k)$.

\item \emph{Pinning} is an operation by which we fixate $x_i$ to a given bit $c$ and generate a new function $f^{x_i=c}(x_1,\ldots,x_{i-1},x_{i+1},\ldots,x_k) = f(x_1,\ldots,x_{i-1},c,x_{i+1},\ldots,x_k)$ for any $x_1,\ldots,x_{i-1},x_{i+1},\ldots,x_k$. It follows that $f\leq^{acyc}_{con}\{f^{x_i=0},f^{x_i=1}\}$.

\item \emph{Projection} is an operation of generating $f^{x_i=*}(x_1,\ldots,x_{i-1},x_{i+1},\ldots,x_k) = \sum_{a\in\{0,1\}} f(x_1,\ldots,x_{i-1},a,x_{i+1},\ldots,x_k)$.

\item \emph{Linking} is an operation by which we replace $x_i$ with $x_j$ and generate $f^{x_i=x_j}(x_1,\ldots,x_{i-1},x_{i+1},\ldots,x_k) = f(x_1,\ldots,x_{i-1},x_j,x_{i+1},\ldots,x_k)$ for any $x_1,\ldots,x_{i-1},x_{i+1},\ldots,x_k$.
\end{enumerate}

For the constraint functions obtained by the above operations, it follows that $arity(\lambda\cdot f)=arity(f)$, $arity(f')=arity(f)+1$, and $arity(f^{x_i=c}) = arity(f^{x_i=*}) = arity(f^{x_i=x_j})= arity(f)-1$.

The linking operation requires a special attention because it does not in general support AT-constructibility. All the above operations except for (5) support $AT$-constructibility.

The following is easy to prove from the definition.

\begin{lemma}\label{T-con-transitivity}
Let $f$ and $g$ be any two constraints and let $\GG$ and $\HH$ be any two constraint sets. If $f\leq^{acyc}_{con} \GG\cup\{g\}$ and $g\leq^{acyc}_{con} \HH$, then $f\leq^{acyc}_{con}\GG\cup \HH$.
\end{lemma}

\begin{proof}
Assume that $f\leq^{acyc}_{con} \GG\cup\{g\}$ and $g\leq^{acyc}_{con}\HH$. There exist an acyclic $\GG\cup\{g\}$-realizable constraint hypergraph of $f$ and an acyclic $\HH$-realizable constraint hypergraph of $g$. Let $f(x_1,\ldots,x_d) = \lambda\sum_{y_1,\ldots,y_m} \prod_{\ell} g(z_1,\ldots,z_d) \ell(w_1,\ldots,w_{d'})$, where $\{z_1,\ldots,z_k,w_1,\ldots,w_{k'}\}\subseteq \{y_1,\ldots,y_m\}$ for $k\leq d$ and $k'\leq d'$, and let $g(z_1,\ldots,z_d) = \xi \sum_{y'_1,\ldots,y'_{m'}} \prod_{h} h(w'_1,\ldots,w'_e)$, where $\{w'_1,\ldots,w'_b\}\subseteq \{y'_1,\ldots,y'_{m'}\}$ for $b\leq e$. From these equalities, it follows that $f(x_1,\ldots,x_d) = \lambda \sum (\xi\sum_{y'_1,\ldots,y'_{m'}} \prod_{h} h(w'_1,\ldots,w'_e))\prod_{\ell} \ell(w_1,\ldots,w_{d'})$. The last term is rewritten in the form of $\lambda\xi \sum \sum \prod h(w'_1,\ldots,w'_{e})\ell(w_1,\ldots,w_{d'})$. Hence, we obtain $f\leq^{acyc}_{con} \HH$.
\end{proof}

We simplify the notation $\sharpacsp(\GG\cup\HH)$ as $\sharpacsp(\GG,\HH)$ unless there is any confusion.

\begin{lemma}\label{acyc-sharpacsp}
If $f\leq^{acyc}_{con}\GG$, then $\sharpacsp(f,\UU)\leq^{\dl}  \sharpacsp(\GG,\UU)$.
\end{lemma}

\begin{proof}
Let $F= \sharpacsp(\{f\}\cup\UU)$ and $F'=\sharpacsp(\GG\cup\UU)$.
Take any instance $I=(Var,C)$ of $\sharpacsp(\{f\}\cup\UU)$ and consider all constraints $C_1,C_2,\ldots,C_t$ of $I$ whose constraint functions are exactly $f$.
Since $f\leq^{acyc}_{con}\GG$, there exist a constant $\lambda\in\complex-\{0\}$ and a finite set $B\subseteq\GG$ for which  $f(x_1,x_2,\ldots,x_k) = \lambda\cdot \prod_{g} g(z_1,z_2,\ldots,z_d)$. For each of such $g(z_1,z_2,\ldots,z_d)$'s, we set $D_{d,z}=(g,z)$ with $z=(z_1,z_2,\ldots,z_d)$. For each $C_i=(f,(v_{i_1},v_{i_2},\ldots,v_{i_{k_i}}))$, we replace the occurrence of $C_i$ in $I$ by $D_{g,z'}$, where $z'$ is obtained from $z$ by replacing $x_1,x_2,\ldots,x_k$ with $v_{i_1},v_{i_2},\ldots,v_{i_{k_i}}$, respectively. We write $I'$ to denote the instance obtained from $I$ by the above replacement. It then follows that $F'(I') = F(I)$.
\end{proof}


We can prove the following relationships between $OR_2$ and $NAND_2$.

\begin{lemma}\label{OR-and-NAND}
Let $u_0=[1,-1]$. The following statements hold.

(1) $OR_2\leq^{acyc}_{con}\{NAND_2,u_0\}$. 

(2) $NAND_2\leq^{acyc}_{con} \{OR_2,u_0\}$.
\end{lemma}

\begin{proof}
Let $u_0=[1,-1]$.
(1) Since $OR_2(x,y) = \sum_{z\in\{0,1\}} NAND_2(x,z) NAND_2(y,z) u_0(z)$, we conclude that
$OR_2\leq^{acyc}_{con}\{NAND_2,u_0\}$.

(2) Similarly, since $NAND_2(x,y) = - \sum_{z\in\{0,1\}} OR_2(x,z) OR_2(y,z) u_0(z)$, we obtain
$NAND_2\leq^{acyc}_{con} \{OR_2,u_0\}$.
\end{proof}

We next present close relationships among four binary constraints: $NAND_k$, $OR_k$, $Implies$, and $RImplies$. 
By the equality $NAND_2(x,y) = \sum_{z\in\{0,1\}} XOR(x,z) Implies(y,z)$, it follows that $NAND_2\leq^{acyc}_{con} \{Implies,XOR\}$ but we can do much better.

\begin{lemma}\label{IMPLIES-case}
Let $k\geq2$. The following statements hold with $u_0=[1,-1]$.

(1) $OR_k\leq^{acyc}_{con} \{Implies,u_0\}$.

(2) $NAND_k\leq^{acyc}_{con}\{Implies,u_0\}$.
\end{lemma}

\begin{proof}
Let $u_0=[1,-1]$ and let $k\geq2$. (1) Consider the case of $k=2$. We define $h(x,y) = - \sum_{z\in\{0,1\}} Implies(x,z) Implies(y,z) u_0(z)$. A simple calculation shows that $h=OR_2$. Hence, we obtain $OR_2\leq^{acyc}_{con} \{Implies,u_0\}$.
More generally, for $k\geq3$, it follows that $OR_k(x_1,x_2,\ldots,x_k) = - \sum_{w\in\{0,1\}} u_0(w) ( \prod_{i=1}^{k} Implies(x_i,w))$.
Therefore, we obtain $OR_k\leq^{acyc}_{con} \{Implies,u_0\}$.
We remark that this proof is much simpler than Lemma 6.6(2) of \cite{Yam12a}.

(2) Note that $NAND_2(x,y) = \sum_{z\in\{0,1\}} Implies(z,x) Implies(z,y) u_0(z)$. This concludes that $NAND_2\leq^{acyc}_{con} \{Implies, u_0\}$. This argument can be expanded to a more general case of $k\geq3$ in a way similar to (1).
\end{proof}

Since $Implies\leq^{acyc}_{con} \{OR_2,XOR\}$, Lemma \ref{IMPLIES-case} implies the following.

\begin{corollary}\label{ORk-to-OR3}
Let $u_0=[1,-1]$. The following statements hold.

(1) For any $k\geq4$, $OR_k\leq^{acyc}_{con}\{OR_2,XOR,u_0\}$.

(2) For any $k\geq2$, $NAND_k\leq^{acyc}_{con}\{OR_2,XOR,u_0\}$.
\end{corollary}

\subsection{Case of General Binary Constraints}\label{sec:case-binary}

We have already discussed the AT-constructibility of four typical binary constraints, $OR_2$, $NAND_2$, $XOR$, and $Implies$.
We next focus our attention on more general binary constraints.

We begin with a brief remark that the constraints $f_1=(0,0,a,0)$ and $f_2=(0,a,0,0)$ with $a\neq0$ can be factorized into $f_1(x,y)=a \cdot \Delta_1(x)\Delta_0(y)$ and $f_2(x,y)= a\cdot \Delta_0(x)\Delta_1(y)$. Thus, we obtain $f_1,f_2\leq^{acyc}_{con} \{\Delta_0,\Delta_1\}$.
In what follows, we consider the other types of binary constraints.

\begin{lemma}\label{OR-character}
Let $f=(1,a,0,b)$ with $ab\neq0$. It follows that $OR_2\leq^{acyc}_{con} \{f,u,u'\}$, where $u=[b/a,1]$ and $u'=[-a^2,1]$.
\end{lemma}

\begin{proof}
We define $h(x,y) = \sum_{z\in\{0,1\}} f(x,z)f(y,z) u(x)u(y) u'(z)$. We then obtain $h=(0,b^2,b^2,b^2)$, which equals $b^2\cdot (0,1,1,1) = b^2\cdot OR_2$. From this equality, we conclude that $OR_2\leq^{acyc}_{con}\{f,u,u'\}$, as requested.
\end{proof}

\begin{lemma}\label{OR-NAND}
Let $f=(0,a,b,1)$ with $ab\neq0$, $u=[1,a]$, and $u'=[-1/b^2,1]$. It follows that $NAND_2 \leq^{acyc}_{con} \{f,u,u'\}$ and $OR_2\leq^{acyc}_{con} \{f,u,u',u_0\}$, where $u_0=[1,-1]$.
\end{lemma}

\begin{proof}
We define $h(x,y) = \sum_{z\in\{0,1\}} f(x,z)f(y,z) u(x)u(y)u'(z)$. It then follows that $h=(a^2,a^2,a^2,0)$, which implies $NAND_2= (1/a^2) \cdot h$. Thus, we conclude that $NAND_2 \leq^{acyc}_{con}\{f,u,u'\}$.
Since $OR_2\leq^{acyc}_{con} \{NAND_2,u_0\}$ with $u_0=[1,-1]$ by Lemma \ref{OR-and-NAND}(1), we obtain $OR_2\leq^{acyc}_{con} \{f,u,u',u_0\}$.
\end{proof}

A $k$-ary constraint (function) $f$ is called \emph{degenerate} if there are $k$ unary constraints $g_1,g_2,\ldots,g_k$ satisfying $f(x_1,x_2,\ldots,x_k) = \prod_{i=1}^{k}g_i(x_i)$ for any series of variables $x_1,x_2,\ldots,x_k$. Let $\mathcal{DG}$ denote the set of all degenerate constraints. Obviously, $\UU\subseteq \DG$ holds.

\begin{lemma}\label{nowhere-zero-constraint}
The following statements hold.

(1) Let $f=(1,a,b,c)$ with $abc\neq0$ and $ab\neq\pm c$. It follows that $NAND_2\leq^{acyc}_{con} \{f,u,u',v\}$, where $u=[1,a/(ab+c)]$, $u'=[-(ab+c/a)^2,1]$, and $v=[a^2,-1]$.

(2) Let $f=(1,a,b,-ab)$ with $ab\neq0$. It follows that $NAND_2\leq^{acyc}_{con} \{f,u,u',v\}$, where $u=[1,b]$, $u'=[-1/a^2,1]$, and $v=[ab,-1]$.

(3) Let $f=(1,a,b,ab)$. It follows that $f\in \DG$.
\end{lemma}

\begin{proof}
(1) Let $v=[a^2,-1]$. We define $h(x,y) = \sum_{z\in\{0,1\}} f(x,z)f(y,z) v(z)$. It follows that $h= (ab-c)\cdot (0,a,a,ab+c) = (a^2b^2-c^2)\cdot (0,a/(ab+c),a/(ab+c),1)$ because $ab\neq \pm c$.
By setting $x= a/(ab+c)$, we define $g=(0,x,x,1)$. By Lemma \ref{OR-NAND}, we conclude that $NAND_2\leq^{acyc}_{con} \{g,u,u'\}$, where $u=[1,x]$ and $u'=[-1/x^2,1]$. Since $h=(a^2b^2-c^2)\cdot g$,
Lemma \ref{T-con-transitivity} then implies $NAND_2\leq^{acyc}_{con} \{h,u,u'\}$. Combining this with $h\leq^{acyc}_{con}\{f,v\}$, we obtain $NAND_2\leq^{acyc}_{con}\{f,u,u',v\}$ by Lemma \ref{T-con-transitivity}.

(2) Let $u=[1,b]$, $u'=[-1/a^2,1]$, and $v=[ab,-1]$. We define $h(x,y) = \sum_{z\in\{0,1\}} f(x,z)f(y,z) v(z)$, which implies $h=2a^2b^2\cdot (0,b,a,1)$. Lemma \ref{OR-NAND} implies that $NAND_2\leq^{acyc}_{con} \{h,u,u'\}$.
Since $h\leq^{acyc}_{con}\{f,v\}$, Lemma \ref{T-con-transitivity} leads to the conclusion that $NAND_2\leq^{acyc}_{con} \{f,u,u',v\}$.

(3) If $ab=0$, then $f$ clearly belongs to $\DG$. Assume that $ab\neq0$. In this case, $f\in\DG$ also follows from the fact that $f(x,y) = u(x)v(y)$ with $u=[1,b]$ and $v=[1,a]$.
\end{proof}

\section{Computational Complexity of \#ACSP($\FF$)}\label{sec:complexity-ACSP}

We have introduced in Section \ref{sec:Sharp-ACSP} the notion of \#ACSPs and defined the counting complexity class $\#\mathrm{ACSP}(\FF)$ for a constraint set $\FF$.
We now wonder what is the computational complexity of $\sharpacsp(\FF)$ for various choices of constraint sets $\FF$.
In what follows, we wish to present a series of results concerning the computational complexity of $\sharpacsp(\FF)$ for a set $\FF$ of constraints.

\subsection{Supporting Results}

Following \cite{Yam12a}, we introduce a few useful constraint sets.
Recall that ${\cal U}$ denotes the set of all unary constraints (i.e., constraints of arity $1$). A constraint is in ${\cal ED}$ if it is a product of some of the following functions: unary functions, $EQ_2$, and $XOR$.
Let ${\cal NZ}$ denote the set of all nowhere-zero constraints (i.e., their outcomes never become zero).
It was shown in \cite{Yam12a} that, for any $f\in\NZ$, $f\in \DG$ iff $f\in\ED$.

Recall from \cite{CLX14,Yam12a} that, for any $\FF\subseteq \ED$, $\sharpcsp(\FF)$ belongs to $\fp_{\complex}$, where $\fl_{\complex}$ denotes the complex-number version of $\fl$.
A similar situation occurs for $\sharpacsp(\FF)$.


\begin{proposition}\label{ED-type-case}
For any constraint set $\FF\subseteq \ED$, $\sharpacsp(\FF)$ is in $\fl_{\complex}$.
\end{proposition}

\begin{proof}
Let $I=(Var,C)$ be any instance of $\sharpacsp(\FF)$. Since each constraint $f$ in $\FF$ is a product of some of unary functions, $EQ_2$, and $XOR$, we can decompose $f$ into these components. Without loss of generality, we assume that $\FF$ consists only of $EQ_2$, $XOR$, and unary functions. It thus suffices to solve the case of $\sharpacsp(EQ_2,XOR)$. In what follows, we want to show how to compute $count(I)$.

Let us consider the associated acyclic constraint graph $G_I=(V_I,E_I)$ of $I$.
Generally, $G_I$ forms a forest since there is no cycle in $G_I$. It is possible to focus on each tree in $G_I$  separately. We choose one of these trees and fixate any node of this tree as a root.
Note that, for each leaf, there is a unique simple path from the root.
Since all edges on this path are associated with either $EQ_2$ or $XOR$, a  truth assignment (that makes all constraint functions ``true'' (or $1$)) is  uniquely determined from the assigned value of the root's variable on this path.
For this reason, we can calculate the partial weight (provided by each assignment) of this tree.
This can be done using log space because, as Reingold \cite{Rei08} showed, using only $O(\log{n})$ space,
we can check whether any two vertices are connected or not.

Hence, we can calculate the partial weight of the entire forest from the assigned values of all the root's variables.
\end{proof}


Gottlob \etalc~\cite{GLS01} demonstrated that $\mathrm{ACSP}$ is $\logcfl$-hard under $\dl$-T-reductions by transforming a $\sac^1$ circuit family to $\mathrm{ACSP}$. Here, we argue the $\sharplogcfl$-hardness of the counting problem $\sharpacsp(OR_2,XOR,u_0)$, where $u_0=[1,-1]$.
Notice that $\{OR_2,XOR\}\nsubseteq \ED$.

\begin{proposition}\label{OR-case}
$\#\mathrm{SAC1P} \leq^{\dl} \sharpacsp(OR_2,XOR,u_0)$, where $u_0=[1,-1]$.  Thus, $\sharpacsp(OR_2,XOR,u_0)$ is $\sharplogcfl$-hard under logspace reductions.
\end{proposition}

For the proof of Proposition \ref{OR-case}, we first prove two useful lemmas. Let us introduce another constraint $ONE$, which
satisfies that $ONE(x_1,x_2,\ldots,x_k) =1$ iff exactly one of $x_1,x_2,\ldots,x_k$ is $1$. In particular, when $k=1$, $ONE(x)=x$ follows for any $x\in\{0,1\}$.

\begin{lemma}\label{calculate-ONE}
The following equations hold for each $k\geq1$.

(1) It follows that the value $OR_2(x,EQ(y,x_1,x_2,\ldots,x_k))$ is equal to the product  $\sum_{u,z_1,z_2,\ldots,z_k\in\{0,1\}} OR_3(x,u,x_1) \cdot OR_3(x,y,z_1) XOR(y,u) \cdot \prod_{i=1}^{k-1} ( OR_3(x,z_i,x_{i+1}) OR_3(x,x_i,z_{i+1}) XOR(x_{i+1},z_{i+1}) )$.

(2) It follows that the value $ONE(x_1,x_2,\ldots,x_k)$ is equal to the product  $OR(x_1,\ldots,x_{k}) \cdot \sum_{z_1,z_2,\ldots,z_k\in\{0,1\}} OR_2(z_1,EQ(z_1,x_2,\ldots,x_k)) \cdot \prod_{i=1}^{k-1} OR(x_1,\ldots,x_{i-1},z_i,EQ(z_i,x_{i+1},\ldots,x_k)) XOR(x_i,z_i)$.
\end{lemma}

\begin{proof}
(1) This is proven by a direct calculation.
(2) We first claim that $ONE(x_1,x_2,\ldots,x_k)$ equals $\sum_{z_1\in\{0,1\}} OR_2(z_1,EQ(z_1,x_2,\ldots,x_k)) \cdot OR_2(x_1,ONE(x_2,\ldots,x_k)) \cdot XOR(x_1,z_1)$.
This equality makes it possible to inductively decompose $ONE$ with the nested use of $OR_2$ and $XOR$. 
Moreover, by (1), $OR_2(z_1,EQ(z_1,x_2,\ldots,x_k))$ is further written in terms of $\{OR_3,XOR\}$. 
We then simplify the nested $OR_2$ and $OR_3$ using the more general operator $OR$. 
This leads to the desired consequence.
\end{proof}


We consider the following three functions associated with AND, OR, and NOT. (1) $F_{OR}(x,y,z)=1$ $\Leftrightarrow$ $OR_2(x,y)=z$. (2) $F_{AND}(x,y,z)=1$ $\Leftrightarrow$ $AND_2(x,y)=z$. (3) $F_{NOT}(x,z)=1$ $\Leftrightarrow$ $NOT(x)=z$.

\begin{lemma}\label{circuit-gate}
Let $u_0=[1,-1]$ and $u_1=[-1,1]$.  The following equations hold.

(1) $F_{OR}(x,y,z) = \sum_{w\in\{0,1\}} OR_3(x,y,w) NAND_2(z,w) u_1(z) u_0(w)$.

(2) $F_{AND}(x,y,z) = - \sum_{w\in\{0,1\}} NAND_3(x,y,w) OR_2(z,w) u_0(w)$.

(3) $F_{NOT}(x,z) = XOR(x,z)$.

These equalities can be extended to arbitrary arities more than $3$.
\end{lemma}

\begin{proof}
The lemma can be proven by checking all values of the both sides of each equation.
\end{proof}


Now, let us return to the proof of Proposition \ref{OR-case}.


\vs{-2}
\begin{proofof}{Proposition \ref{OR-case}}
Let $u_0=[1,-1]$. We intend to show the $\sharplogcfl$-hardness of $\sharpacsp(OR_2,XOR,u_0)$ under logspace reductions. Recall from Lemma \ref{SAC1P-complete} that $\#\mathrm{SAC1P}$ is $\sharplogcfl$-complete.
By Corollary \ref{ORk-to-OR3}, $OR_3\leq^{acyc}_{con}\{OR_2,XOR,u_0\}$ follows.
This implies that $\sharpacsp(OR_3,OR_2,XOR)\leq^{\dl} \sharpacsp(OR_2,XOR,u_0)$
It thus suffices to prove that $\#\mathrm{SAC1P}$ is logspace reducible to $\sharpacsp(OR_3,OR_2,XOR)$.

Let us consider an instance $\pair{C,x}$ given to \#SAC1P, where $C$ is a leveled semi-unbounded circuit of depth at most $\log{n}$ with $n$ input bits and $x\in\{0,1\}^n$.
Let $C$ have size $n^{O(1)}$ and depth $d$ with $1\leq d \leq \log{n}$. We first construct a \emph{skeleton tree} $T_{skt}$ from $C$ by removing the ``labels'' from all gates (including input gates), trimming all but one child node from each OR gate, and keeping intact the two child nodes of each AND gate. We then assign new labels $(i,j)$ to all nodes $v$ of $T_{skt}$ as follows. The root has label $(d,1)$. If node $v$ is located at level $i$ and is the parent of two children labeled by $(i-1,2j-1)$ and $(i-1,2j)$, then the label of $v$ is $(i,j)$. If $v$ has only one child with label $(i-1,j)$, then its parent has label $(i,j)$.

We wish to construct a constraint set by ``translating'' each gate into an associated constraint as follows.

(1) Consider an OR gate labeled $(i,j)$ with $m$ children whose labels are $(i-1,j_1),(i-1,j_2),\ldots,(i-1,j_m)$.
Let $\boldvec{x}_{i-1}= (x_{i-1,j_1},\ldots,x_{i-1,j_m})$. 
Let us recall the constraint $ONE$. 
We then set $F_{ONE}(\boldvec{x}_{i-1},x_{i,j})$ to be $\sum_{z_{i,j}\in\{0,1\}} g(\boldvec{x}_{i-1}, x_{i,j},z_{i,j}) XOR(x_{i,j},z_{i,j})$, where $g(\boldvec{x}_{i-1}, x_{i,j},z_{i,j})$ denotes the function $\sum_{z_{i,j}\in\{0,1\}} OR_2(z_{i,j},ONE(\boldvec{x}_{i-1}))  \cdot OR_2(x_{i,j},EQ(x_{i,j},\boldvec{x}_{i-1}))$. It is possible to make $F_{ONE}$ expressed in terms of $\{OR_3,OR_2,XOR\}$.
This is done by Lemmas \ref{calculate-ONE} and \ref{circuit-gate}.

(2) Consider an AND gate labeled $(i,j)$ with two children $(i-1,j_1)$ and $(i-1,j_2)$. We translate it into $F_{AND}(x_{i-1,j_1},x_{i-1,j_2},x_{i,j})$.  Since $EQ_2\leq^{acyc}_{con} \{OR_2,XOR\}$, $F_{AND}$ is expressed in terms of $\{OR_2,XOR\}$.

(3) Each input gate that takes the $j$th bit of $x$
is translated into the basic variable $x_j$.

Let $V$ denote the set of all variables $x_{i,j}$ and let $U$ be the set of all constraints generated above. Let $U=\{f_i\}_{i\in[t]}$ for $t\in\nat^{+}$. Notice that all $f_i$'s are in $\{OR_3,OR_2,XOR\}$.
Using (1)--(3), $\pair{C,x}$ is translated into the instance  $I_{C,x}=(V,\{(f_i,(v_{i_1},\ldots,v_{i_k}))\}_{i})$.
Since the above procedure can be carried out using log space, $\#\mathrm{SAC1P}$ is
logspace
reducible to $\sharpacsp(OR_3,OR_2,XOR)$, as requested.
\end{proofof}


A \emph{2CNF (Boolean) formula} $\phi$ is of the form $(z_{11}\vee z_{12})\wedge (z_{21}\vee z_{22}) \wedge \cdots \wedge (z_{n1}\vee z_{n2})$, where each $z_{ij}$ is a literal over a set $\{x_1,x_2,\ldots,x_n\}$ of variables. Such a formula $\phi$ is said to be \emph{acyclic} if its associated constraint hypergraph $G_{\phi}$ is acyclic.
The \emph{counting acyclic 2CNF satisfiability problem} (abbreviated $\#\mathrm{Acyc\mbox{-}2SAT}$) is the problem of counting the total number of satisfying assignments of a given acyclic 2CNF formula.
In comparison, it is known in  \cite{Val79} that $\#\mathrm{SAT}$ (counting satisfiability problem) is $\sharpp$-complete.

\begin{lemma}\label{Acyc-2SAT-vs-Implies}
$\#\mathrm{Acyc}\mbox{-}\mathrm{2SAT} \leq^{\dl} \sharpacsp(Implies,u_0)$, where $u_0=[1,-1]$.
\end{lemma}

\begin{proof}
Let $u_0=[1,-1]$. Consider any acyclic 2CNF Boolean formula $\phi$. Let $\phi\equiv (z_{11}\vee z_{12})\wedge (z_{21}\vee z_{22}) \wedge \cdots \wedge (z_{n1}\vee z_{n2})$, where each item $z_{ij}$ is a literal. We wish to translate $\phi$ into a certain instance $I$ of $\sharpacsp(Implies,u_0)$.
Let $Var$ denote the set $\{x_{j}\mid j\in[m]\}$ of all variables appearing in $\phi$.
We define the set $C=\{C_{i}\}_{i\in[n]}$ of constraints by setting $C_i$ to be $C_i = (OR_2, (x_{j_1},x_{j_2}))$ if $z_{i1}=x_{j_1}$ and $z_{i2}=x_{j_2}$,
$C_i=(Implies, (x_{j_1},x_{j_2}))$ if $z_{i1}=\overline{x_{j_1}}$ and $z_{i2}=x_{j_2}$,
$C_i=(Implies, (x_{j_2},x_{j_1}))$ if $z_{i1}=x_{j_1}$ and $z_{i2}= \overline{x_{j_2}}$, and
$C_i=(NAND_2, (x_{j_1},x_{j_2}))$ if $z_{i1} = \overline{x_{j_1}}$ and $z_{i2} = \overline{x_{j_2}}$. Symbolically, we write $C_i=(f_i,(x_{j_1},x_{j_2}))$, where $f_i$, $x_{j_1}$, and $x_{j_2}$ are defined above. The pair $(Var,C)$ then forms the desired instance $I$.

It is not difficult to show that, given an assignment $\sigma$, the partial weight $\prod_{i=1}^{n} f_i(\sigma(x_{j_1}),\sigma(x_{j_2}))$ equals $1$ iff $\sigma$ satisfies $\phi$. Hence the number of satisfying assignments of $\phi$ equals $count(I)$.
Since $OR_2$ and $NAND_2$ are AT-constructible from $\{Implies,u_0\}$ by  Lemma \ref{IMPLIES-case}, Lemma \ref{acyc-sharpacsp} immediately concludes  that $I$ belongs to $\sharpacsp(Implies,u_0)$.
\end{proof}


Let us consider the converse of Lemma \ref{Acyc-2SAT-vs-Implies}.
Unfortunately, we do not know if $\sharpacsp(Implies,u_0)$ is logspace reducible to $\#\mathrm{Acyc\mbox{-}2SAT}$; however, we can still prove a slightly weaker statement. For this purpose, we introduce a restricted variant of $\UU$. Let $\UU^{(-)}$ denote the set $\{ [0,1], [1,0], [0,0], [1,1]\}$.

\begin{lemma}
$\sharpacsp(Implies,\UU^{(-)}) \leq^{\dl} \#\mathrm{Acyc}\mbox{-}\mathrm{2SAT}$.
\end{lemma}

\begin{proof}
Let $I=(Var,C)$ be any instance of $\sharpacsp(Implies,\UU^{(-)})$. Let $G_I=(V_I,E_I)$ denote the associated hypergraph.
Note that $count(I) = \sum_{\sigma} [ \prod_{i\in[s]}Implies(\sigma(v_{i1}),\sigma(v_{i2})) \cdot \prod_{i\in[4]} f_i(\sigma(w_{j_i})) ]$, where $U=\{f_1,f_2,f_3,f_4\}$.
We convert $I$ to a 2CNF formula as follows. If $(Implies,(u,v))\in C$, then we include $\overline{u}\vee v$ as a clause. Assume that $C$ contains $(f,u)$ for a unary constraint $f$. If $f=[1,0]$, then we add $\overline{u}$. If $f=[0,1]$, then we add $u$. If $f=[0,0]$, then we add two clauses $u$ and $\overline{u}$. If $f=[1,1]$, then we add $u\vee \overline{u}$.
It is not difficult to show that $\prod_{i\in[t]}Implies(\sigma(v_{i1}),\sigma(v_{i2})) \cdot \prod_{i\in[4]} f_i(\sigma(w_{j_i}))=1$ iff the obtained formula is true.
The above conversion can be carried out using log space.
\end{proof}

Let $IMP$ be composed of all constraints that are logically equivalent to conjunctions of a finite number of constraints of the form $\Delta_0$, $\Delta_1$, and $Implies$. We say that a constraint $f$ has \emph{imp support} if the support $\supp(f)$ of $f$ belongs to $IMP$.
For example, $AND_k$ belongs to $\ED$ because $AND_k$ can be expressed as a product of $EQ_k$ and $\Delta_1$: $AND_k(x_1,x_2,\ldots,x_k) = EQ_k(x_1,x_2,\ldots,x_k) \Delta_1(x_1)$.


\begin{lemma}\label{notin-ED-OR2}
If $f\notin \ED$, then $\sharpacsp(OR_2)\leq^{\dl} \sharpacsp(f)$.
\end{lemma}

\begin{proof}
Let $arity(f)=k$.  We show the lemma by induction on $k$. Consider the base case of $k=2$. Since $f\notin\ED$, $f$ must have the form specified in Lemmas \ref{OR-character}--\ref{nowhere-zero-constraint}. Moreover, by Lemma \ref{OR-and-NAND}(2), $NAND_2$ is $AT$-constructed from $\{OR_2,u_0\}$, where $u_0=[1,-1]$.

Consider the general case of $k\geq3$. By induction hypothesis, the lemma is true for $k-1$.
We consider the case where $f$ has imp support. There are two cases to examine.

(a) Consider the case where $f$ is in $(IMP\cap \ED)\cup\{EQ_1\}$. In this case, we have a contradiction with $f\notin \ED$.
(b) Consider the case where $f$ is not in $(IMP\cap \ED)\cup\{EQ_1\}$. Similar to \cite[Lemma 7.4]{Yam12a}, we can show the following claim. For any $f\notin (IMP\cap \ED)\cup\{EQ_1\}$, there are two constraints $h\in (IMP\cap \ED)\cup\{EQ_1\}$ of arity $m'\leq k$ and $g$ of arity $k-m$ with $k\neq m\leq m'$ such that $f(x_1,\ldots,x_k) = h(x_1,\ldots,x_{m'}) g(x_m,\ldots,x_k)$ (after appropriate re-ordering of variables) and $g\leq^{acyc}_{con} f$. By induction hypothesis, it follows that $\sharpacsp(OR_2)\leq^{\dl} \sharpacsp(g)$. This implies $\sharpacsp(OR_2)\leq^{\dl} \sharpacsp(f)$.
\end{proof}

\subsection{Classification by the Free Use of XOR}

In this section, we allow free use of $XOR$, which can enhance the computational strength of $\sharpacsp$ when it is given with other appropriate constraints.


Let us return to the first main theorem (Theorem \ref{main-theorem}) stated in Section \ref{sec:main-contribution} and provide its proof.
Here, we rephrase this main theorem as follows.

\bs
\n{\bf Theorem \ref{main-theorem} (rephrased)}\hs{2}
For any set $\FF$ of Boolean constraints, if $\FF\subseteq \ED$, then $\sharpacsp(\FF,\UU,XOR)$ is in $\fl_{\complex}$. Otherwise, $\sharpacsp(\FF,\UU,XOR)$ is $\sharplogcfl$-hard under logspace reductions.
\bs


The proof of this theorem utilizes some of the key results of \cite{Yam12a}, in particular, Lemma 7.5, Lemma 8.3, Corollary 8.6, and Proposition 8.7.
For our convenience, we rephrase the crucial part of \cite[Lemma 7.5]{Yam12a} as follows.

\begin{lemma}\label{Yam12a-lemma}
For any constraint function $f\in\mathcal{NZ}-\mathcal{DG}$ of arity $k\geq2$, there exist $k-1$ constants $c_3,c_4,\ldots,c_k\in\{0,1\}^{k-2}$ and $\lambda\in\complex-\{0\}$ that make the constraint function $h= \lambda\cdot f^{x_3=c_3,x_4=c_4,\ldots,x_k=c_k}$ have the form $(1,x,y,z)$ with $xyz\neq0$ and $xy\neq z$ after an appropriate permutation of variable indices.
\end{lemma}

Since the constraint $h$ given in the above lemma is obtained from $f$ by ``pinning'' and ``normalization'', we conclude that $h\leq^{acyc}_{con}\{f\}\cup\UU$.

In contrast to $IMP$, the notation $\mathcal{IM}$ expresses the set of all functions, not in $\mathcal{NZ}$, which are products of a finite number of
unary functions and $Implies$.


\vs{-2}
\begin{proofof}{Theorem \ref{main-theorem}}
Let $\FF$ denote any set of constraint functions.
We first remark that $\sharpacsp(\FF,\UU,XOR)$ is in $\sharplogcfl_{\complex}$ by Theorem \ref{falling-sharplogcfl}. Let us consider the case of $\FF\subseteq \mathcal{ED}$. In this case, Lemma \ref{ED-type-case} concludes that $\sharpacsp(\FF,\UU,XOR)$ belongs to $\fl_{\complex}$. We next consider the case of $\FF\nsubseteq \mathcal{ED}$ and choose a constraint $f\in\FF$ not in $\mathcal{ED}$. Note that $art(f)\geq2$ because, otherwise, $f$ falls into $\UU$ ($\subseteq \mathcal{ED}$), a contradiction.

Since $\sharpacsp(OR_2,XOR,u_0)$ with $u_0=[1,-1]$ is $\sharplogcfl$-hard by Proposition \ref{OR-case}, if $\sharpacsp(OR_2)\leq^{\dl} \sharpacsp(\FF,\UU,XOR)$, then we conclude that $\sharpacsp(\FF,\UU,XOR)$ is also $\sharplogcfl$-hard.  Our goal is thus to show that $\sharpacsp(OR_2)\leq^{\dl} \sharpacsp(\FF,\UU,XOR)$.

Let us examine two possible cases (1)--(2).

(1) Assume that $f\in\mathcal{NZ}$. Note that $f\notin\mathcal{DG}$ because, otherwise, $f\in\mathcal{DG}\subseteq \mathcal{ED}$ follows, a contradiction.  Assume that $f$ has arity $2$. This $f$ has the form $(1,a,b,c)$ with $abc\neq0$. Since $f\notin \mathcal{DG}$, we further obtain $ab\neq c$ because, otherwise, $f=(1,a,b,ab)=[1,b]\cdot [1,a] \in\mathcal{DG}$.
Lemma \ref{nowhere-zero-constraint} concludes that $NAND_2\leq^{acyc}_{con} \{f\}\cup\UU$.
Lemma \ref{OR-and-NAND}(1) then helps us obtain $OR_2\leq^{acyc}_{con} \{f\}\cup\UU$.
By Lemma \ref{acyc-sharpacsp}, we it follows that $\sharpacsp(OR_2) \leq^{\dl} \sharpacsp(f,\UU,XOR)$.

Next, we assume that $f$ has arity $k>2$. Apply Lemma \ref{Yam12a-lemma}. By the remark after Lemma \ref{Yam12a-lemma}, we obtain $h=(1,x,y,z)$ such that  $xyz\neq0$, $xy\neq z$, and $h\leq^{acyc}_{con}\{f\}\cup\UU$. By Lemma \ref{T-con-transitivity}, we conclude that  $NAND_2\leq^{acyc}_{con}\{f\}\cup\UU$. Similarly to the previous case, we can obtain the desired conclusion.

(2) Assume that $f\notin\mathcal{NZ}$.
(a) Let us consider the case where $f\notin\mathcal{IM}$. We begin with the case of $art(f)=2$. Let $f=(a,b,c,d)$ with $a,b,c,d\in\complex$. By \cite[Lemma 8.3]{Yam12a}, $f\notin\mathcal{IM}$ implies that $ad=0$ and $bc\neq0$. Since $f\notin\mathcal{ED}$, this is not the case of $a=d=0$.
In the case of $art(f)\geq3$, by applying an appropriate series of pinning and linking to $f$ as in \cite[Proposition 8.7]{Yam12a}, we obtain $g=(a,b,c,0)$ with $abc\neq0$ from $f$. Therefore, we conclude that $g\leq^{acyc}_{con}\{f\}\cup\UU$.
(b) Assume that $f\in\mathcal{IM}$. By \cite[Corollary 8.6]{Yam12a}, after conducting pinning to $f$, we obtain $h$, which is in $\mathcal{IM}-\mathcal{ED}$. Thus, we conclude that $h\leq^{acyc}_{con}\{f\}\cup\UU$.

Overall, using the results in Sections \ref{sec:acyclic-const}--\ref{sec:case-binary}, we obtain the theorem.
\end{proofof}

\subsection{Classification without Free Use of XOR}

Next, let us recall the statement of the second main theorem (Theorem \ref{second-theorem}) from Section \ref{sec:main-contribution}. We rephrase it as follows.

\bs
\n{\bf Theorem \ref{second-theorem} (rephrased)}\hs{2}
For any set $\FF$ of constraint functions, if $\FF\subseteq \ED$, then $\sharpacsp(\FF,\UU)$ is in $\fl_{\complex}$. Otherwise, if $\FF\subseteq \IM$, then  $\sharpacsp(\FF,\UU)$ is hard for  $\#\mathrm{Acyc\mbox{-}2SAT}$ under logspace reductions. Otherwise, $\sharpacsp(\FF,\UU)$ is $\sharplogcfl$-hard under logspace reductions.
\bs


Hereafter, we provide the proof of this main theorem. For the proof, we need the following two supporting lemmas: Lemmas \ref{implies-bound} and \ref{IM-OR-XOR}.

\begin{lemma}\label{implies-bound}
For any $f\in\IM - \DG\cup \NZ$, $ \#\mathrm{Acyc}\mbox{-}\mathrm{2SAT}  \leq^{\dl} \sharpacsp(f) \leq^{\dl} \sharpacsp(Implies,u_0)$, where $u_0=[1,-1]$.
\end{lemma}

\begin{proof}
We intend to prove the lemma by induction on the value of $arity(f)$. Let us begin with the base case of $arity(f)=2$. Since $f\notin \DG\cup\NZ$, $f$ has the form $(1,a,0,b)$ with $ab\neq0$, $(0,a,b,1)$ with $ab\neq0$, $(1,a,b,c)$ with $abc\neq0$ and $ab\neq\pm c$, and $(1,a,b,-ab)$ with $ab\neq0$. By Lemmas \ref{OR-character}--\ref{nowhere-zero-constraint}, we obtain either $OR_2\leq^{acyc}_{con}\{f,u,u',u''\}$ or $NAND_2\leq^{acyc}_{con}\{f,u,u',u''\}$ for appropriately chosen unary constraints $u,u',u''$. By Lemma \ref{OR-and-NAND}, the latter AT-constructibility relation is replaced by $OR_2\leq^{acyc}_{con}\{f,u,u',u''\}$.

Case: $arity(f)>2$. Note that $f$ is obtained by multiplying $Implies$ and unary constraints. By pinning variables of $f$ except for two variables, we AT-construct another constraint, say, $g$ from $\{f\}\cup \UU$. Since $g\leq^{acyc}_{con}\{f\}\cup\UU$, by Lemma \ref{acyc-sharpacsp}, we obtain $\sharpacsp(g,\UU)\leq^{\dl} \sharpacsp(f,\UU)$. If, for any choice of pinning operations, $g$ belongs to $\DG\cup\NZ$, then we can conclude that $f$ is also in $\DG\cup\NZ$, a contradiction. Hence, there is a series of appropriate pinning operations such that $g$ is in $\IM-\DG\cup\NZ$. Since $arity(g)=2$, this case can be handled as in the base case.
\end{proof}

\begin{lemma}\label{IM-OR-XOR}
If $f\notin \IM\cup\ED$, then $\sharpacsp(OR_2,XOR) \leq^{\dl} \sharpacsp(f)$.
\end{lemma}


With the help of Lemmas \ref{implies-bound} and  \ref{IM-OR-XOR},  Theorem \ref{second-theorem} can be proven in the following fashion.

\vs{-2}
\begin{proofof}{Theorem \ref{second-theorem}}
Let us recall the proof of Theorem \ref{main-theorem}. It therefore suffices to discuss the case where $\FF\nsubseteq \ED$.
Consider the case where $\FF\subseteq \IM$. By Lemma \ref{implies-bound}, $\sharpacsp(\FF,\UU)$ is hard for $\#\mathrm{Acyc\mbox{-}2SAT}$. On the contrary, if $\FF\nsubseteq \IM$, then Lemma \ref{IM-OR-XOR} leads to the conclusion  that $\sharpacsp(f,\UU)$ is hard for $\sharpacsp(OR_2,\UU,XOR)$.
\end{proofof}

\section{Brief Concluding Discussions}

The past literature has explored various situations of $\sharpcsp$ and it has presented numerous complete characterizations of $\sharpcsp$s when restricted to those situations.

This work has initiated a study of counting acyclic satisfaction problems (abbreviated as $\sharpacsp$s). The computational complexity of these problems depend on the types of constraints given as part of inputs.
As the main contributions of this work, we have presented two complexity classifications of all complex-weighted \#ACSPs (Theorems \ref{main-theorem}--\ref{second-theorem}), depending on the free use of XOR.
To prove them, we have developed a new technical tool, called AT-constructibility.

Unfortunately, there still remain numerous questions that have not yet been discussed or answered throughout this work. We hope that future research will resolve these questions to promote our basic understanding of counting CSPs (or \#CSPs).

\renewcommand{\labelitemi}{$\circ$}
\begin{enumerate}\vs{-1}
  \setlength{\topsep}{-2mm}%
  \setlength{\itemsep}{1mm}%
  \setlength{\parskip}{0cm}%

\item In this work, we have used complex-weighted constraints. How does the complexity classification of any weighted \#ACSP look like if the weights are limited to real numbers or even nonnegative rational numbers? The negative values of constraints, for instance, makes it possible to simplify the complexity classification of $\sharpacsp$s.

\item Does the use of randomized computation instead of deterministic computation significantly alter the complexity classification of $\sharpacsp$? What if we use quantum computation?

\item What if we place a restriction on the maximum degree of constraints?

\item At this moment, we do not know that $\#\mathrm{Acyc\mbox{-}2SAT}$ is $\sharplogcfl$-complete (under logspace reductions). We conjecture that this situation may not happen.
\end{enumerate}

\let\oldbibliography\thebibliography
\renewcommand{\thebibliography}[1]{%
  \oldbibliography{#1}%
  \setlength{\itemsep}{-2pt}%
}
\bibliographystyle{plainurl}
\bibliographystyle{alpha}

\end{document}